\documentclass[11pt]{scrartcl}

\usepackage{url}
\usepackage[hidelinks]{hyperref}
\usepackage[utf8]{inputenc}
\usepackage[small]{caption}
\usepackage{graphicx}
\usepackage{amsmath}
\usepackage{booktabs}
\usepackage{amsthm}
\usepackage{amssymb}

\usepackage{cleveref}
\usepackage{color}
\usepackage{xspace}
\usepackage{tikz}
\usepackage{tabularx}
\usepackage{pbox}
\usepackage{framed}
\usepackage[shortlabels]{enumitem}

\usepackage{natbib}

\usepackage{complexity}
\usepackage{thm-restate}
\usepackage{algorithm}
\usepackage{algorithmic}

\newtheorem{theorem}{Theorem}[section]

\newtheorem{proposition}[theorem]{Proposition}

\theoremstyle{definition}
\newtheorem{definition}[theorem]{Definition}

\newtheorem{claim}{Claim}

\definecolor{myblue}{RGB}{0,82,147}
\definecolor{myorange}{RGB}{227,114,34}
\definecolor{mygreen}{RGB}{162,173,0}
\definecolor{mygray}{HTML}{909090}

\tikzset{ 
protovertex/.style={
  draw,
  circle,
  inner sep=0,
  minimum size=.2cm}
}

\tikzset{ 
bprotovertex/.style={
  draw = myblue,
  fill = myblue,
  circle,
  inner sep=0,
  minimum size=.2cm}
}

\tikzset{ 
gprotovertex/.style={
  draw = mygreen,
  fill = mygreen,
  circle,
  inner sep=0,
  minimum size=.2cm}
}

\usetikzlibrary{arrows.meta}
\tikzset{>={Latex[width=1.8mm,length=1.8mm]}}
\usetikzlibrary{decorations.pathmorphing,shapes,snakes}
\usetikzlibrary{calc}
\usetikzlibrary{arrows, decorations.markings}
\usetikzlibrary{positioning}

\newcommand{\tdg}{TDG\xspace}
\newcommand{\asgnm}{\lambda\xspace}
\newcommand{\partasgnm}{\nu\xspace}
\allowdisplaybreaks

\title{Topological Distance Games\footnote{A preliminary version of this paper appears in Proceedings of the 37th AAAI Conference on Artificial Intelligence \citep{BuSu23a}.
Most of this research was done when the first author was a PhD student at Technical University of Munich.}}

\author{
Martin Bullinger\\University of Oxford
\and
Warut Suksompong\\National University of Singapore
}

\date{\vspace{-5ex}}

\begin{document}

\maketitle

\begin{abstract}
We introduce a class of strategic games in which agents are assigned to nodes of a topology graph and the utility of an agent depends on both the agent's inherent utilities for other agents as well as her distance from these agents on the topology graph.
This model of \emph{topological distance games (TDGs)} offers an appealing combination of important aspects of several prominent settings in coalition formation, including (additively separable) hedonic games, social distance games, and Schelling games.
We study the existence and complexity of stable outcomes in TDGs---for instance, while a jump stable assignment may not exist in general, we show that the existence is guaranteed in several special cases.
We also investigate the dynamics induced by performing beneficial jumps.
\newline

\end{abstract}

\section{Introduction}

You arrive at a hotel for your organization's annual banquet, and some of the seats at the tables have already been taken.
You would like to sit close to your friends who work in the same team or share similar hobbies.
On the other hand, you want to stay away from colleagues whom you had unpleasant interactions with lately.
Which seat should you take?
Once everyone has picked a seat, would you regret not having chosen a different seat?
Similar issues arise when assigning faculty members to offices in the department building, students to desks in a classroom, or employees to cottages at a company retreat.

Recently, \citet{BMM22a} introduced the model of \emph{hedonic games with fixed-size coalitions}, wherein the agents are to be partitioned into coalitions whose sizes have been determined in advance, for example, by the sizes of the tables at the banquet.
They assumed \emph{additively separable} utilities, meaning that the utility of an agent for a coalition is the sum of her utilities for the individual agents in her coalition.
While their model partially captures some of the aforementioned scenarios---for instance, each table at the banquet can be considered as one coalition---it neglects an important aspect common in such scenarios: the agents are typically assigned to specific locations, and agents prefer to be located close to their friends and far from their enemies.
In our banquet scenario, a person sitting next to you has a higher influence on your utility than someone at the opposite end of the table, as you are much more likely to engage in a conversation with the former person than the latter.
Similarly, you will in all likelihood run into your office neighbor more frequently than you encounter your colleague at the other end of the corridor.

With these motivating examples in mind, we introduce a class of games that we call \emph{topological distance games (TDGs)}.
An instance of TDG contains a topology graph, which is an undirected graph that specifies the locations to which the agents can be assigned.
The influence that an agent~$i$ has on another agent~$j$ is $j$'s inherent utility for~$i$ scaled by a factor depending on the distance between the two agents on the topology graph; if the two agents are not connected on the graph, they have no influence on each other.
TDGs combine important aspects of several well-studied coalition formation settings, including hedonic games, social distance games, and Schelling games---we discuss these connections in detail in \Cref{sec:related}.
We sometimes assume that the scaling factor is the reciprocal of the distance between the two agents, but most of our results also hold for arbitrary (strictly decreasing) distance factor functions.
Following additively separable hedonic games (ASHGs), we then take the utility of an agent for an assignment to be the sum of all other agents' influences on the agent in question.
Our formal model is described in \Cref{sec:prelim}.

\subsection{Our Results}

We study a fundamental notion of stability in our setting---\emph{jump stability}---which requires that no agent would rather jump to some empty node than stay at her current node.
In \Cref{sec:symmetric}, we warm up by considering the case where agents' utilities are symmetric, that is, for any pair of agents $i$ and $j$, $i$'s inherent utility for $j$ is the same as $j$'s inherent utility for $i$.
We show that for any distance factor function, there exists a jump stable assignment; on the other hand, finding such an assignment is a \PLS-complete problem.

In \Cref{sec:asymmetric}, we investigate the more general setting where the utilities are not necessarily symmetric.
We observe that a jump stable assignment may no longer exist, even when there are only two agents.
On the other hand, if utilities are non-negative and the friendship graph\footnote{That is, the directed graph indicating (ordered) pairs of agents $i,j$ such that $i$'s inherent utility for $j$ is positive.} is acyclic, then existence is guaranteed.
We then focus on the case where the topology graph is a cycle and
every vertex in the friendship graph has out-degree at most $1$.
In this case, we characterize the friendship graphs for which the resulting instance admits a jump stable assignment, and present an efficient algorithm for computing such an assignment for those graphs.
We also provide existence and non-existence results when the topology graph is a path or an (extended) star, and show that deciding the existence is \NP-hard for the \emph{reciprocal distance factor function}, whereby the scaling factor is the reciprocal of the distance between the two agents.

Lastly, in \Cref{sec:dynamics}, we explore dynamical aspects of TDGs.
If utilities are non-negative and the
friendship graph is acyclic,
we show that the jump dynamics is guaranteed to converge; however, even under these restrictions, the dynamics may run for an exponential number of steps.
In addition, we establish the \NP-hardness of deciding if the dynamics can possibly converge, or if it necessarily converges.

\subsection{Related Work}
\label{sec:related}

The model of TDGs shares certain similarities with a number of existing models, and therefore offers an appealing combination of important aspects of several prominent settings in coalition formation.
\Cref{fig:enter-label} illustrates some relations between TDGs and other settings.

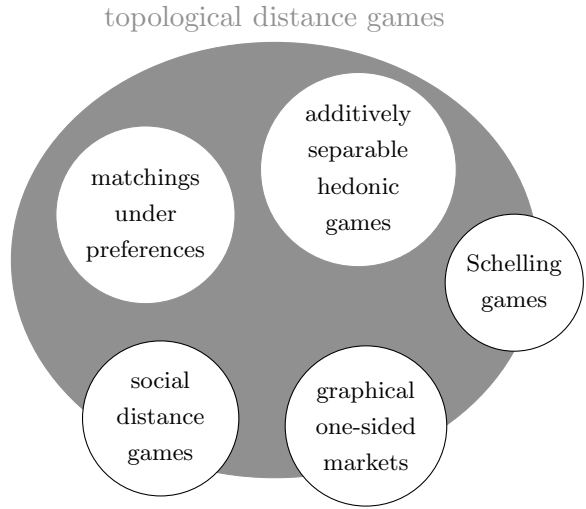
\begin{figure}
    \centering
\begin{tikzpicture}
	\draw[fill, mygray, draw = mygray] (0,0) ellipse (9em and 7.5em);
	\node[mygray] at (0,3.2) {{topological distance games}};
    \node[align = center,draw = mygray,circle, fill = white] at (-1.7,.6) {\footnotesize{matchings}\\ \footnotesize{under}\\ \footnotesize{preferences}};
    \node[align = center,draw = mygray,circle, fill = white] at (1.1,1.2) {\footnotesize{additively}\\ \footnotesize{separable}\\\footnotesize{hedonic}\\ \footnotesize{games}};
    \node[align = center,draw,circle, fill = white] at (1.2,-2.2) {\footnotesize{graphical}\\ \footnotesize{one-sided}\\ \footnotesize{markets}};
	\node[align = center,draw,circle, fill = white] at (-1.5,-2.1) {\footnotesize{social}\\ \footnotesize{distance}\\ \footnotesize{games}};
	\node[align = center,draw,circle, fill = white] at (3.15,-.3) {\footnotesize{Schelling}\\ \footnotesize{games}};
\end{tikzpicture}
    \caption{Topological distance games in relation to other cooperative games. Additively separable hedonic games and matchings under preferences are subclasses. Social distance games share the idea of a distance-based utility decay, and graphical one-sided markets and Schelling games also assume the existence of a topology graph.}
    \label{fig:enter-label}
\end{figure}

Firstly, TDGs are similar to the aforementioned \emph{hedonic games with fixed-size coalitions} \citep{BMM22a} in that one could view each connected component of the topology graph as a coalition of fixed size.
In formal terms, hedonic games with fixed-size coalitions form a subclass of TDGs where every connected component is a clique.
The main difference between TDGs and hedonic games in general \citep{AzSa16a} is that hedonic games do not come with a topology graph, so only the partition of the agents into coalitions matters, whereas in TDGs the distances resulting from the assignment of agents to the topology graph can affect the level of influence that the agents have on one another.
Note that additively separable utilities are commonly studied in hedonic games \citep{BoJa02a,ABS13a}.
In fact, ASHGs can be viewed as a special case of TDGs where the topology graph consists of $n$ cliques of size $n$ each ($n$ denotes the number of agents).

Secondly, \citet{BrLa11a} proposed the class of \emph{social distance games}, wherein there is a social network that captures the connections among agents.
The agents are again partitioned into coalitions, but now the utility of an agent for a coalition is the average, over all agents in the coalition, of the reciprocals of the distance to each agent in the coalition.
Here, the distance is taken with respect to the subgraph of the social network induced by the coalition, and the distance of an agent to herself is not considered in this calculation.
Like in hedonic games, there is no topology graph in social distance games.
\citet{FKOV21a} introduced \emph{distance hedonic games}, which generalize social distance games by allowing the distance function to be arbitrary rather than specifically the reciprocal function.
\citet{ReRe22a} considered a distance-based approach for extending the agents' preferences over neighbors to preferences over coalitions in a subclass of hedonic games.

Thirdly, \citet{MaSi19a} studied \emph{graphical one-sided markets}, which assume the existence of both an agent graph and a topology graph.
Agents are placed on the topology graph, and the utility of an agent for the placement is the sum of the agent's utility for her assigned node and her utilities for her neighbors on the topology graph, where the latter utilities are taken from the agent graph.
While our TDG model is more restrictive from the point of view that it does not allow agents to derive utilities from nodes in the topology graph, it is more general in that an agent's utility depends not only on neighboring agents on the topology graph, but also on agents further away.
\citet{EPTZ20a} investigated a similar model as \citet{MaSi19a} from the truthfulness perspective, while \citet{BHJO20a} considered the special case where an agent's utility depends on her neighbors but not on her assigned node.

Fourthly, a recent stream of work on \emph{Schelling games} also deals with a topology graph, and an agent's utility depends on the neighboring agents on the graph \citep{CLM18a,EFLM19a,AEGI21a,BSV21a}.
However, in that model, the agents have predetermined types and the utility of an agent is defined as the \emph{fraction} of neighboring agents of the same type.
Note that jump stability is commonly studied in Schelling games as well.

Finally, \emph{matchings under preferences} \citep{GaSh62a} can also be modelled in the framework of TDGs if we take a topology graph that consists of disjoint edges (and therefore has many connected components). However, the classical notion of stability for matching instances would correspond to group deviations by two agents. 

\section{Preliminaries}\label{sec:prelim}

Let $N = [n]$ be the set of \emph{agents}, where $[k] := \{1,2,\dots,k\}$ for each positive integer~$k$.
Each agent~$i\in N$ is endowed with an \emph{(inherent) utility function} $u_i\colon N\to\mathbb{R}$, which specifies the inherent utility that $i$ has for every other agent; we assume that $u_i(i) = 0$ for all $i$.
A utility function $u_i$ is
\begin{itemize}
\item \emph{symmetric} if $u_i(j) = u_j(i)$ for all $i,j\in N$, and
\item \emph{binary} if $u_i(j)\in\{0,1\}$ for all $i,j\in N$.
\end{itemize}
We say that agent~$j$ is a \emph{friend} of agent~$i$ if $u_i(j) > 0$; the \emph{friendship graph}\footnote{Friendship graphs have been studied in several papers on hedonic games \citep{IOSY19a,KLRR20a,BuKo21a}.} is a directed graph with the set of nodes $N$ such that there is an edge from $i$ to $j$ if and only if $u_i(j) > 0$.

There is a \emph{topology graph} $G = (V, E)$, which is a simple (not necessarily connected) undirected graph with at least $n$~nodes.
An \emph{assignment} $\asgnm\colon N\to V$ is an injective mapping that assigns each agent to a node in $V$, i.e., each node can be occupied by at most one agent.
For $N'\subseteq N$, let $\asgnm(N') := \left\{\asgnm(i)\mid i\in N'\right\}$.
A node $v\in V$ is called \emph{empty} with respect to an assignment $\asgnm$ if $v\notin \asgnm(N)$.

The \emph{distance factor function} $f\colon \mathbb{Z}_{\ge 1}\to \mathbb{R}_{> 0}$ is a strictly decreasing function which determines the level of influence that an agent has on another agent depending on the distance between them, where this distance is the length of the shortest path between their assigned nodes in $G$.
If the two agents are assigned to different connected components of $G$, then the distance factor between them is taken to be $0$, meaning that they have no influence on each other's utilities.
The \emph{reciprocal distance factor function} refers to the function\footnote{As discussed in \Cref{sec:related}, a similar idea involving reciprocals of the distance has been used in social distance games \citep{BrLa11a}.} $f(k) = 1/k$.
We abuse notation slightly by extending utility functions~$u_i$ to assignments.
In particular, the utility of agent~$i$ for assignment~$\asgnm$ is
\[
u_i(\asgnm) := \sum_{j\in N\setminus\{i\}} f(d_G(\asgnm(i), \asgnm(j)))\cdot u_i(j),
\]
where $d_G(v,v')$ denotes the length of a shortest path between $v$ and $v'$ in~$G$.
A \emph{topological distance game (TDG)} consists of the agents and their utility functions, the topology graph, and the distance factor function.

Given an assignment $\asgnm$, an empty node $v\in V$, and an agent $i\in N$, denote by $\asgnm^{i\to v}$ the assignment that results when $i$ jumps from her assigned node $\asgnm(i)$ to $v$.
We now define the main stability notion that we study in this paper.
\begin{definition}
Given an instance of TDG and an assignment~$\asgnm$, a jump by agent~$i$ to an empty node~$v$ is a \emph{beneficial jump} in~$\asgnm$ if $u_i(\asgnm) < u_i(\asgnm^{i\to v})$.

The assignment~$\asgnm$ is said to be \emph{jump stable} if no agent has a beneficial jump, that is, for each agent $i\in N$ and each empty node $v\in V$, it holds that $u_i(\asgnm) \ge u_i(\asgnm^{i\to v})$.

\end{definition}
Since we consider jump stability, we assume without loss of generality that $|V| > n$, as every assignment is trivially jump stable if $|V| = n$.

\section{Warm-Up: Symmetric Utilities}
\label{sec:symmetric}

We begin by deriving preliminary results for the case of symmetric utilities.
First, by a potential function argument, we can show the existence of a jump stable assignment in this case. 
This idea is common in the literature on (additively separable) hedonic games---it was first used by \citet{BoJa02a}, and applied again by
\citet{Suks15a} and \citet{BMM22a}.

\begin{theorem}
\label{thm:symmetric-jump-existence}
For any distance factor function and symmetric utilities, there exists a jump stable assignment.
\end{theorem}

\begin{proof}
Consider an assignment maximizing the potential function $\Phi(\asgnm) := \sum_{i\in N}u_i(\asgnm)$; such an assignment must exist because the number of possible assignments is finite.
Assume for contradiction that some agent~$i^*$ prefers to jump from her current node~$v$ to an empty node~$w$.
We have 
\begin{align}
u_{i^*}(\asgnm) < u_{i^*}(\asgnm^{i^*\to w}), 
\label{eq:potential-jump}
\end{align}
where we know that
\[
u_{i^*}(\asgnm) = \sum_{j\in N\setminus\{i^*\}} f(d_G(v, \asgnm(j)))\cdot u_{i^*}(j)
\]
and
\[
u_{i^*}(\asgnm^{i^*\to w}) = \sum_{j\in N\setminus\{i^*\}} f(d_G(w, \asgnm(j)))\cdot u_{i^*}(j).
\]
Now, if $i^*$ jumps to~$w$, the potential function changes by 
\begin{align*}
 \Phi(\asgnm^{i^*\to w}) - \Phi(\asgnm) &= \sum_{i\in N}u_i(\asgnm^{i^*\to w}) - \sum_{i\in N}u_i(\asgnm) \\
&= \left[u_{i^*}(\asgnm^{i^*\to w}) + \sum_{j\in N\setminus\{i^*\}} \left( f(d_G(\asgnm(j), w))\cdot u_j(i^*) \right) \right] \\
&\qquad - \left[u_{i^*}(\asgnm) + \sum_{j\in N\setminus\{i^*\}} \left( f(d_G(\asgnm(j), v))\cdot u_j(i^*) \right) \right] \\
&= \left[u_{i^*}(\asgnm^{i^*\to w}) + \sum_{j\in N\setminus\{i^*\}} \left( f(d_G(w,\asgnm(j)))\cdot u_{i^*}(j) \right) \right] \\
&\qquad - \left[u_{i^*}(\asgnm) + \sum_{j\in N\setminus\{i^*\}} \left( f(d_G(v,\asgnm(j)))\cdot u_{i^*}(j) \right) \right] \\ 
&= 2\cdot \left( u_{i^*}(\asgnm^{i^*\to w}) - u_{i^*}(\asgnm) \right) > 0,
\end{align*}
where we use the symmetry of the utilities for the second equality and \eqref{eq:potential-jump} for the inequality. 
This contradicts the assumption that $\asgnm$ maximizes the potential function $\Phi$.
\end{proof}

Despite its guaranteed existence, a jump stable assignment can be difficult to compute.
Our reduction is from a local variant of \textsc{Max Cut}.

\begin{restatable}{theorem}{symmetricjumpPLS}
\label{thm:symmetric-jump-PLS}
For any distance factor function and symmetric utilities, finding a jump stable assignment is \PLS-complete.
\end{restatable}

\begin{proof}
Membership in \PLS\ is clear, as we can consider the potential $\Phi(\asgnm) = \sum_{i\in N}u_i(\asgnm)$ from the proof of \Cref{thm:symmetric-jump-existence} as the measure to be maximized. Checking for a local improvement only requires the inspection of a polynomial number of jumps.

For \PLS-hardness, we provide a reduction from the \PLS-complete problem \textsc{Max Cut} under the \textsc{Flip}-neighborhood \citep{ScYa91a}. 
An instance of \textsc{Max Cut} consists of a complete and undirected weighted graph.
A solution is a partition of the vertices into two subsets, where the \emph{cut} between the two subsets, i.e., the total weight of the edges between vertices of the two subsets, is to be (locally) maximized. 
Given a solution of an instance of \textsc{Max Cut}, its \textsc{Flip}-neighborhood contains all partitions in which exactly one vertex changes the subset to which it belongs.

Consider an instance $H = (V_P,E_P,w)$ of \textsc{Max Cut} where $V_P$, $E_P$, and $w\colon E_P\to \mathbb R$ are the vertex set, the edge set, and the weight function, respectively. 
Let $t := |V_P|$.
We define the reduced \tdg as follows. 
Let $N = \{\alpha_v\colon v\in V_P\}$ be the set of agents. 
The utility functions are given by $u_{\alpha_x}(\alpha_y) = - w(\{x,y\})$, where $x,y\in V_P$; clearly, this defines a symmetric utility function. 
The topology graph $G = (V,E)$ is given by $V = A\cup B$ where $A = \{a_i\colon i\in [t]\}$, $B = \{b_i\colon i\in [t]\}$, and $E = \{\{a_i,a_j\},\{b_i,b_j\}\colon 1\le i < j\le t\}$. 
In other words, the topology graph consists of two cliques of size $t$.

Now, every assignment $\asgnm$ induces the $2$-partition $P_{\asgnm} = (A_P, B_P)$ with $A_P = \{x\in V_P\colon \alpha_x\in \asgnm^{-1}(A)\}$ and $B_P = \{x\in V_P\colon \alpha_x \in \asgnm^{-1}(B)\}$.

Consider a vertex $x\in A_P$ and an empty node $v\in B$ with respect to $\asgnm$.
The change in the weight of the cut when $x$ changes its partition class, that is, the change in weight when going from the cut induced by $P_{\asgnm}$ to the cut induced by $P_{\asgnm^{\alpha_x\to v}}$, is exactly 
\begin{align*}
    \sum_{y\in A_P\setminus \{x\}} w(x,y) &- \sum_{y\in B_P} w(x,y)
    = -u_{\alpha_x}(\asgnm) - (- u_{\alpha_x}(\asgnm^{\alpha_x\to v}))\text.
\end{align*}
A similar computation holds if $x\in B_P$.

Hence, by the computations in \Cref{thm:symmetric-jump-existence}, the change in the value of a cut after $x\in V_P$ switches its subset in $P_{\asgnm}$ is exactly half of the change of the potential $\Phi$ when going from $\asgnm$ to $\asgnm^{\alpha_x\rightarrow v}$, where $v$ is a node in the component different from $\asgnm(\alpha_x)$'s component of $G$.
Moreover, $\Phi$ does not change when $\alpha_x$ jumps to another node in the same component.
Therefore, $P_{\asgnm}$ is locally optimal if and only if $\asgnm$ is.
\end{proof}

\section{Asymmetric Utilities}
\label{sec:asymmetric}

We now consider the more general setting where the agents' inherent utilities for each other are not necessarily symmetric.
First, we observe that a jump stable assignment may no longer exist, even when there are only two agents.

\begin{proposition}
\label{prop:cat-and-mouse}
Let $G$ be a connected graph of diameter at least $3$.
For any distance factor function, there exists an instance with topology graph $G$ and two agents such that no jump stable allocation exists.
\end{proposition}

\begin{proof}
Let $n = 2$, $u_1(2) = 1$, and $u_2(1) = -1$, and consider any assignment~$\asgnm$. 
If $\mathit{d_G}(\asgnm(1),\asgnm(2))\ge 2$, then agent~$1$ would jump to a neighboring node of agent~$2$. 
Else, $\mathit{d_G}(\asgnm(1),\asgnm(2)) = 1$.
In this case, consider two nodes $v,w\in V$ with $\mathit{d_G}(v,w)\ge 3$. 
It must be that $\mathit{d_G}(\asgnm(1),v)\ge 2$ or $\mathit{d_G}(\asgnm(1),w)\ge 2$.
If the former holds, then $v\ne \asgnm(2)$, and agent~$2$ has a beneficial jump to $v$; an analogous argument holds in the latter case with $w$ instead of $v$.
\end{proof}

\Cref{prop:cat-and-mouse} does not hold if we lower the diameter threshold to $2$: for a star graph and any number of agents, one can check that there is always a jump stable assignment.

As we will see later, a jump stable assignment may not exist even if utilities are non-negative and the friendship graph is a cycle.
We show next that the existence of such an assignment is guaranteed under non-negative utilities when the friendship graph is acyclic; this result will also be useful for our characterization when the friendship graph is a cycle or, more generally, every vertex in the friendship graph has out-degree at most $1$ (\Cref{thm:cycle-in-cycle}).
We remark that acyclic friendship graphs can model situations in which there is a hierarchy among the agents: for instance, at research conferences, younger researchers may be keen to interact with certain older ones who might help their career, but not vice versa.

\begin{theorem}\label{thm:DAGassignment}
    For any distance factor function, if the friendship graph is acyclic and utilities are non-negative, then a jump stable assignment exists and can be computed in polynomial time.
\end{theorem}

\begin{proof}
    Assume that the friendship graph is acyclic and utilities are non-negative. 
    There exists a \emph{topological order} $\pi$ of the agents in $N$ \citep{Kahn62a}, i.e., a function $\pi\colon N\to [n]$ such that $\pi(i) > \pi(j)$ whenever $u_i(j) > 0$. Based on this order, we provide a polynomial-time algorithm that computes a jump stable assignment. 
    Given a subset of agents $M\subseteq N$, a partial assignment $\asgnm\colon M \to V$, an agent $i\in N\setminus M$, and a node $v\in V\setminus \asgnm(M)$, let $\asgnm^{[i\to v]}\colon M\cup \{i\}\to V$ be the partial assignment extended from $\asgnm$ by assigning agent~$i$ to $v$.
    Also, let $\asgnm_{\emptyset}\colon \emptyset \to V$ be the empty partial assignment.

\begin{algorithm}

  \caption{Jump stable assignment for acyclic friendship graph and non-negative utilities.}
  \label{alg:jumpacyclic}
  \textbf{Input:} Topology graph $G = (V,E)$, topological order $\pi\colon N \to [n]$\\
  \textbf{Output:} Jump stable assignment $\asgnm\colon N\to V$

  \begin{algorithmic}[]
\STATE $V^e\leftarrow V$
\STATE $\asgnm_0 \leftarrow \asgnm_\emptyset$

  \FOR{$k = 1,\dots, n$}
    \STATE $i\leftarrow \pi^{-1}(k)$
    \STATE Select $v^* \in \arg\max_{v\in V^e}\{u_i(\asgnm_{k-1}^{[i\to v]})\}$
    \STATE $\asgnm_k \leftarrow \asgnm_{k-1}^{[i\to v^*]}$
    \STATE $V^e \leftarrow V^e \setminus \{v^*\}$
  \ENDFOR
  \RETURN $\asgnm \leftarrow \asgnm_n$
 \end{algorithmic}
\end{algorithm}

    \Cref{alg:jumpacyclic} describes our procedure for computing a jump stable assignment. 
    In each iteration of the for-loop, we determine the position of some agent. 
    We place agents according to the topological order $\pi$---this ensures that whenever an agent is placed, all of her friends have already been placed, so her utility is not influenced by later agents.
    
    Clearly, the algorithm runs in polynomial time. 
    It remains to show that the returned assignment~$\asgnm$ is jump stable. 
    Consider an arbitrary agent $i\in N$; it suffices to show that $i$ cannot perform a beneficial jump.
    Notice that all nodes that $i$ could potentially jump to were available at the moment when $i$ was assigned during the execution of \Cref{alg:jumpacyclic}. 
    Moreover, since $u_i(j) = 0$ for all agents~$j$ with $\pi(j) > \pi(i)$, we have $u_i(\asgnm^{i\to v}) = u_i(\asgnm_{\pi(i)}^{i\to v}) \le u_i(\asgnm_{\pi(i)}) = u_i(\asgnm)$ for any node $v\in V\setminus \asgnm(N)$; here, the inequality follows from the maximization in the algorithm. 
    Hence, agent~$i$ cannot perform a beneficial jump.
\end{proof}

If the topology graph is a cycle (e.g., a party table), utilities are non-negative, and every vertex in the friendship graph has out-degree at most~$1$, we completely characterize the friendship graphs for which the resulting instance admits a jump stable assignment.

\begin{restatable}{theorem}{cycleincycle}
\label{thm:cycle-in-cycle}
    Suppose that the topology graph~$G$ is a cycle, and each agent has at most one friend and utility $0$ for the remaining agents.
    For any distance factor function, a jump stable assignment exists if and only if neither of the following cases occurs:
    \begin{itemize}
        \item The friendship graph is a $3$-cycle;
        \item The friendship graph is a $5$-cycle.
    \end{itemize}
    If a jump stable assignment exists, it can be computed in polynomial time.
\end{restatable}

\begin{proof}
Consider a \tdg such that $G$ is a cycle and each agent has at most one friend and utility $0$ for the remaining agents.

First, suppose that the friendship graph is either a $3$-cycle or a $5$-cycle.
In particular, $n \in \{3,5\}$.
Without loss of generality, let $u_i(i+1) > 0$ for $i\in [n]$, where we take agent indices modulo $n$.
Recall from \Cref{sec:prelim} that $|V| > |N|$.
Assume for contradiction that $\asgnm \colon N\to V$ is a jump stable assignment, and consider an empty node $v\in V\setminus \asgnm(N)$. 
By following the cycle $G$ starting from $v$ in both directions, at some point, we hit the first occupied node in each direction. 
Suppose that agents $j,k\in [n]$ are assigned to these nodes. 
Note that $d_G(\asgnm(j),\asgnm(k))\ge 2$, because there is either an empty node or a node occupied by another agent on each of the two paths connecting $\asgnm(j)$ and $\asgnm(k)$. 
If $n = 3$, we can assume without loss of generality that $k = j+1$, in which case $j$ has a beneficial jump to an empty neighboring node of $k$, a contradiction. 
Consider now the case $n = 5$. 
Then, because each of $j$ and $k$ has an empty neighboring node, the other neighboring node must be occupied by $j-1$ and $k-1$, respectively. 
In particular, the set $\{j,j-1,k,k-1\}$ contains four pairwise distinct agents. 
Assume without loss of generality that $j+1 = k-1$. 
Observe that $d_G(\asgnm(j),\asgnm(k-1))\ge 3$---indeed, one of the paths connecting $\asgnm(j)$ and $\asgnm(k-1)$ has $v$ and $\asgnm(k)$ as intermediate nodes, while the other path has $\asgnm(j-1)$ and $\asgnm(j-2)$ as intermediate nodes. 
Hence, $j$ has a beneficial jump to the empty neighboring node of $k$, a contradiction.

From now on, assume that the friendship graph is neither a $3$-cycle nor a $5$-cycle. 
The existence of a jump stable assignment follows by combining a few observations which we capture in claims.

\begin{claim}\label{cl:longcycles}
Let $M\subseteq N$ be a set of agents that form a cycle of length $m\ge 2$ with $m\not\in\{3,5\}$ in the friendship graph, and let $\{c_1,\dots, c_m\}\subseteq V$ be a set of $m$ consecutive nodes in the cycle $G$. 
Then, there exists a partial assignment $\partasgnm\colon M\to \{c_1,\dots, c_m\}$ such that, regardless of how the agents in $N\setminus M$ are assigned, no agent in~$M$ can perform a beneficial jump.
\end{claim}
\renewcommand\qedsymbol{$\vartriangleleft$} 
\begin{proof}
Let $M$ and $\{c_1,\dots, c_m\}$ be as in the statement of the claim. 
Suppose without loss of generality that $M = [m]$ and $u_i(i+1) > 0$ for $i\in [m]$, where we take agent indices modulo $m$.

Assume first that $m$ is even. Consider the assignment
\begin{align*}
    \partasgnm(i) = \begin{cases}
       c_{m/2 - i + 1} & i\le m/2;\\
       c_i & i\ge m/2 + 1.
    \end{cases}
\end{align*}
See \Cref{fig:exlongcycle} for an illustration.
Then, for $i\notin\{m/2,m\}$, $\partasgnm(i)$ is adjacent to $\partasgnm(i+1)$, so no beneficial jump is possible. 
For $i\in \{m/2,m\}$, it holds that $d_G(\partasgnm(i),\partasgnm(i+1)) = m/2$, and any jump to an empty node cannot decrease this distance.

\begin{figure}
    \centering
    \begin{tikzpicture}[
	element/.style={shape=circle,draw, fill=white,scale=0.8}
	]
        \node (c1) at (0,0) {};
        \node[element, label = {[xshift = 0.05cm, yshift = 0.1cm]180:$c_1$}] (v11) at ($(c1)+(90:2)$) {\color{white}0};
        \node[element, label = {[xshift = 0.05cm, yshift = -0.1cm]180:$c_6$}] (v16) at ($(c1)+(270:2)$) {\color{white}0};
        \foreach \i in {2,3,4,5}{
        \pgfmathsetmacro{\rot}{54 + \i*36}
        \node[element, label = {[xshift = 0.05cm]180:$c_\i$}] (v1\i) at ($(c1)+(\rot:2)$) {\color{white}0};}
        \node (w11) at (v11) {3};
        \node (w12) at (v12) {2};
        \node (w13) at (v13) {1};
        \node (w14) at (v14) {4};
        \node (w15) at (v15) {5};
        \node (w16) at (v16) {6};
        \foreach[count = \j] \i in {2,3,4,5,6}
        \draw (v1\j) edge (v1\i);
        \draw (v11) edge ($(v11) + (0.5,0)$);
        \draw (v16) edge ($(v16) + (0.5,0)$);
        \node at ($(v11) + (0.8,0)$) {\dots};
        \node at ($(v16) + (0.8,0)$) {\dots};
        
        \node (c2) at (4.5,0) {};
        \node[element, label = {[xshift = 0.05cm, yshift = 0.1cm]180:$c_1$}] (v11) at ($(c2)+(90:2.3)$) {\color{white}0};
        \node[element, label = {[xshift = 0.05cm, yshift = -0.1cm]180:$c_7$}] (v17) at ($(c2)+(270:2.3)$) {\color{white}0};
        \foreach \i in {2,3,4,5,6}{
        \pgfmathsetmacro{\rot}{60 + \i*30}
        \node[element, label = {[xshift = 0.05cm]180:$c_\i$}] (v1\i) at ($(c2)+(\rot:2.3)$) {\color{white}0};}
        \node (w11) at (v11) {1};
        \node (w12) at (v12) {7};
        \node (w13) at (v13) {2};
        \node (w14) at (v14) {6};
        \node (w15) at (v15) {3};
        \node (w16) at (v16) {4};
        \node (w17) at (v17) {5};
        \foreach[count = \j] \i in {2,3,4,5,6,7}
        \draw (v1\j) edge (v1\i);
        \draw (v11) edge ($(v11) + (0.5,0)$);
        \draw (v17) edge ($(v17) + (0.5,0)$);
        \node at ($(v11) + (0.8,0)$) {\dots};
        \node at ($(v17) + (0.8,0)$) {\dots};
    \end{tikzpicture}
    \caption{Partial assignments in the proof of \Cref{cl:longcycles} for $m=6$ (left) and $m = 7$ (right).}
    \label{fig:exlongcycle}
\end{figure}
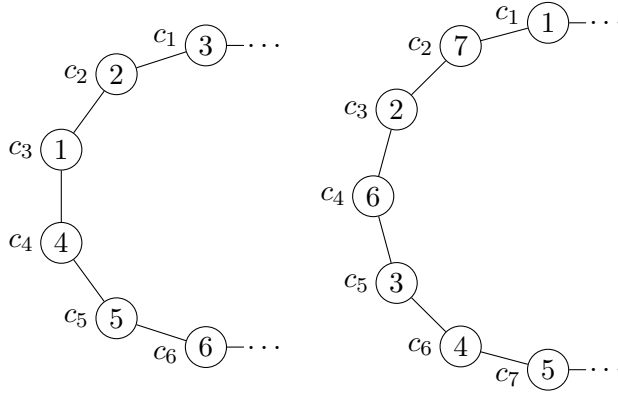

Assume now that $m$ is odd. Consider the assignment
\begin{align*}
    \partasgnm(i) = \begin{cases}
       c_{2i - 1} & i\le \lfloor m/2\rfloor;\\
       c_{m-1} & i = \lfloor m/2\rfloor + 1;\\
       c_m & i = \lfloor m/2\rfloor + 2;\\
       c_{2(m+1-i)} & i \ge \lfloor m/2\rfloor + 3.
    \end{cases}
\end{align*}
Again, see \Cref{fig:exlongcycle} for an illustration.
The claim follows from a case distinction:
\begin{itemize}
    \item For $i\le \lfloor m/2\rfloor - 1$, it holds that $d_G(\partasgnm(i),\partasgnm(i+1)) = 2$, and no jump can lower this distance.
    \item For $i\in \{\lfloor m/2\rfloor, \lfloor m/2\rfloor + 1 , m\}$,  it holds that $d_G(\partasgnm(i),\partasgnm(i+1)) = 1$, and no jump can lower this distance.
    \item For $i = \lfloor m/2\rfloor + 2$, it holds that $\partasgnm(i) = c_m$ and $\partasgnm(i+1) = c_{m - 3}$. The closest that $i$ can get to $i+1$ is by jumping to the possibly empty node next to $c_1$. However, as $m\ge 7$, this jump leads to a distance of at least $4$ from $c_{m-3}$, which is not beneficial.
    \item For $\lfloor m/2\rfloor + 3\le i \le m - 1$, it holds that $d_G(\partasgnm(i),\partasgnm(i+1)) = 2$, and no jump can lower this distance.
\end{itemize}
This establishes the claim.
\end{proof}

The previous claim helps us place agents with a cyclic friendship structure unless the cycle is of length $3$ or $5$. 
For these short cycles, we need to assign some other agent in addition to those on the cycle at the same time.

\begin{claim}\label{cl:shortcycles}
Let $M\subseteq N$ be a set of agents that form a cycle of length $m\in \{3,5\}$ in the friendship graph, and let $\{c_1,\dots, c_{m+1}\}\subseteq V$ be a set of $m+1$ consecutive nodes in the cycle $G$. 
Let $a\in N\setminus M$. 
Then, there exists a partial assignment $\partasgnm\colon M\cup\{a\}\to \{c_1,\dots, c_{m+1}\}$ such that, regardless of how the agents in $N\setminus (M\cup\{a\})$ are assigned, no agent in $M$ can perform a beneficial jump.
\end{claim}

\begin{proof}
Let $M$, $a$, and $\{c_1,\dots, c_{m+1}\}$ be as in the statement of the claim. 
Suppose without loss of generality that $M = [m]$ and $u_i(i+1) > 0$ for $i\in [m]$, where we take agent indices modulo $m$. Then, for the partial assignment indicated in \Cref{fig:exshortcycle}, no agent in $M$ can perform a beneficial jump.
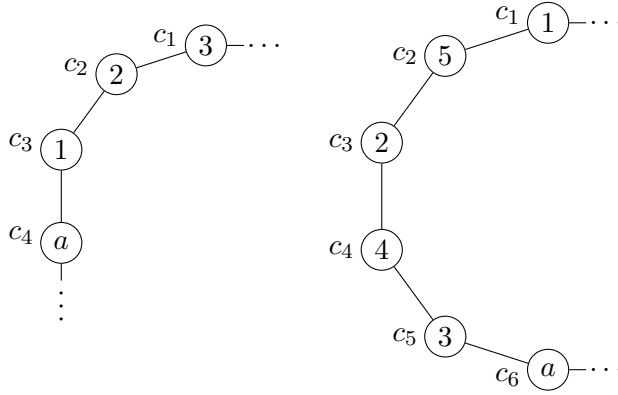
\begin{figure}
    \centering
    \begin{tikzpicture}[
	element/.style={shape=circle,draw, fill=white, scale = 0.8}
	]
        \node (c1) at (0,0) {};
        \foreach \i in {1,2,3,4}{
        \pgfmathsetmacro{\rot}{54 + \i*36}
        \node[element, label = {[xshift = 0.05cm, yshift = 0.1cm]180:$c_\i$}] (v1\i) at ($(c1)+(\rot:2)$) {\color{white}0};}
        \node (w11) at (v11) {3};
        \node (w12) at (v12) {2};
        \node (w13) at (v13) {1};
        \node (w14) at (v14) {$a$};
        \foreach[count = \j] \i in {2,3,4}
        \draw (v1\j) edge (v1\i);
        \draw (v11) edge ($(v11) + (0.5,0)$);
        \draw (v14) edge ($(v14) + (0,-0.5)$);
        \node at ($(v11) + (0.8,0)$) {\dots};
        \node[rotate = 90] at ($(v14) + (0,-0.8)$) {\dots};
        
        \node (c2) at (4.5,0) {};
        \node[element, label = {[xshift = 0.05cm, yshift = 0.1cm]180:$c_1$}] (v11) at ($(c2)+(90:2.3)$) {\color{white}0};
        \node[element, label = {[xshift = 0.05cm, yshift = -0.1cm]180:$c_6$}] (v16) at ($(c2)+(270:2.3)$) {\color{white}0};
        \foreach \i in {2,3,4,5}{
        \pgfmathsetmacro{\rot}{54 + \i*36}
        \node[element, label = {[xshift = 0.05cm]180:$c_\i$}] (v1\i) at ($(c2)+(\rot:2.3)$) {\color{white}0};}
        \node (w11) at (v11) {1};
        \node (w12) at (v12) {5};
        \node (w13) at (v13) {2};
        \node (w14) at (v14) {4};
        \node (w15) at (v15) {3};
        \node (w16) at (v16) {$a$};
        \foreach[count = \j] \i in {2,3,4,5,6}
        \draw (v1\j) edge (v1\i);
        \draw (v11) edge ($(v11) + (0.5,0)$);
        \draw (v16) edge ($(v16) + (0.5,0)$);
        \node at ($(v11) + (0.8,0)$) {\dots};
        \node at ($(v16) + (0.8,0)$) {\dots};
    \end{tikzpicture}
    \caption{Partial assignments in the proof of \Cref{cl:shortcycles} for cycles of length $3$ (left) and $5$ (right).}
    \label{fig:exshortcycle}
\end{figure}
\end{proof}

With the two claims in hand, we are ready to find a desired jump stable assignment. 
The agent set $N$ can be partitioned into the set $C$ of agents that lie on some directed cycle in the friendship graph and the set $A$ of remaining agents; note that $A$ induces a directed acyclic subgraph of the friendship graph. 

First, assume that $C$ contains more than one cycle. 
If one of the cycles is of length $3$ or $5$---call that $C_1$---then assign the corresponding agents as in \Cref{cl:shortcycles} to an arbitrary set of consecutive nodes in $G$. 
Then, place the next cycle, $C_2$, on an adjacent set of consecutive nodes starting with the node occupied by agent $a$ for $C_1$ in \Cref{fig:exshortcycle}. 
If $C_2$ is also of length $3$ or $5$, we place it in such a way that the node corresponding to agent $a$ for~$C_2$ in \Cref{fig:exshortcycle} is already occupied by an agent in $C_1$. 
If $C_2$ is of any other length, we use the assignment of \Cref{cl:longcycles}. 
Similarly, we assign the remaining agents in $C$, cycle by cycle, to sets of consecutive nodes adjacent to already assigned nodes, according to the assignments in  Claims~\ref{cl:longcycles} and \ref{cl:shortcycles}, respectively, taking care that the critical neighboring node of a cycle of length $3$ and $5$ is already occupied. 
By Claims~\ref{cl:longcycles} and \ref{cl:shortcycles}, no agent in $C$ can perform a beneficial jump.
We then assign the agents in $A$ as in \Cref{thm:DAGassignment} by placing them one by one in an optimal position after their friend (if any) has been placed. 
This yields a jump stable assignment.

Next, if there is only one cycle in $C$, we can apply the same procedure unless the cycle is of length $3$ or $5$. 
In this case, since we assume that the friendship graph is neither a $3$-cycle nor a $5$-cycle, $A$ is nonempty.
We first place an agent $a\in A$ that has no friend in $A$; such an agent exists due to the acyclicity of the friendship graph induced by the agents in $A$. 
Then, we place the cycle next to this agent according to the assignment in \Cref{cl:shortcycles}. 
If $a$ has a friend in $C$, then we use a permutation of this cycle such that $a$ is on the node adjacent to her friend. 
Afterwards, we place the remaining agents in $A$ as in the previous paragraph. 
This yields a jump stable assignment.

Finally, if $C$ is empty, we can simply apply \Cref{thm:DAGassignment}.

In all cases, our argument directly yields a polynomial-time algorithm for computing a jump stable assignment.
\renewcommand\qedsymbol{$\square$}\end{proof}

We have seen that, with non-negative utilities, a jump stable assignment always exists if the friendship graph is acyclic, but may not exist if both the friendship graph and the topology graph are cycles.
Is existence guaranteed when the friendship graph is a cycle but the topology graph is acyclic?
If the topology graph is a path, the answer is positive as long as each agent has at most two friends and non-negative utilities.

\begin{theorem}
\label{thm:path-two-friends}
For any distance factor function, if the topology graph~$G$ is a path and each agent has at most two friends and utility $0$ for the remaining agents, then a jump stable assignment exists and can be computed in polynomial time.
\end{theorem}

\begin{proof}
We assign agents to nodes along the path from left to right.
First, assign an arbitrary agent to the leftmost node.
For each subsequent node, among the unassigned friends~$j$ of the agent~$i$ occupying the previous node, assign one with the highest $u_i(j)$; if all of $i$'s friends have been assigned (or if $i$ has no friend), then assign an arbitrary agent.
Clearly, this procedure runs in polynomial time.

We show that the final assignment $\asgnm$ is jump stable.
Consider any agent~$i$.
Since all utilities are non-negative, it suffices to show that the jump to the $(n+1)$th node from the left, denoted by $v$, is not beneficial for $i$.
If all of $i$'s friends are to her left, this is obvious.
Else, if $i$ has one friend~$j$ to her right, then $j$ is directly next to $i$, so by jumping to the $(n+1)$th node, $i$ gets closer to neither $j$ nor $i$'s other friend (if $i$ has a friend other than $j$).
Otherwise, $i$ has both friends $j$ and $k$ to her right.
Assume without loss of generality that $u_i(j) \ge u_i(k)$ and that the algorithm assigns $j$ next to $i$.
Suppose that $i$ and $k$ are on the $d_i$-th and $d_k$-th nodes from the left, respectively (so $j$ occupies the $(d_i + 1)$th node).
Then we have
\begin{align*}
u_i(\asgnm) 
&= f(1)\cdot u_i(j) + f(d_k - d_i)\cdot u_i(k) \\
&\ge f(1)\cdot u_i(k) + f(d_k - d_i)\cdot u_i(j) \\
&\ge f(n+1-d_k)\cdot u_i(k) + f(n - d_i)\cdot u_i(j)= u_i(\asgnm^{i\to v}),
\end{align*}
where the second inequality holds because $d_k\le n$ and $f$~is a decreasing function.
\end{proof}

On the other hand, if the topology graph is a tree, existence is no longer guaranteed even when the friendship graph is a cycle.

\begin{theorem}
\label{thm:cyclicfs:treetop}
For any distance factor function, a jump stable assignment may not exist even if the topology graph is a tree, the agents have binary utilities, and the friendship graph is a cycle.
\end{theorem}

\begin{proof}
Consider an instance with $N = [6]$ and the topology graph $G = (V,E)$ depicted in \Cref{fig:topoex:treetop}.
Assume that the utilities are given by $u_{i}(i+1) = 1$ for $i\in [6]$, where we take agent indices modulo $6$, and $u_i(j) = 0$ for all other pairs $i,j\in N$.
Suppose for contradiction that there exists a jump stable assignment $\asgnm\colon N \to V$. 
Let $W := \asgnm(N)\subseteq V$ be the set of occupied nodes. 

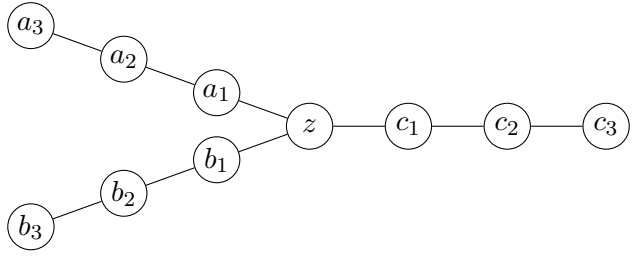
\begin{figure}
    \centering	\begin{tikzpicture}[
	element/.style={shape=circle,draw, fill=white,scale =0.9}
	]
		\node[element] (z) at (0,0) {\color{white}0};
		\node at (z) {$z$};
		\foreach[count = \k] \i in {1.3,2.6,3.9}{
		\node[element] (a1\k) at (160:\i) {\color{white}0};
		\node (a\k) at (a1\k) {$a_{\k}$};
		\node[element] (b1\k) at (200:\i) {\color{white}0};
		\node (b\k) at (b1\k) {$b_{\k}$};
		\node[element] (c1\k) at (0:\i) {\color{white}0};
		\node (c\k) at (c1\k) {$c_{\k}$};
		}
		\foreach \i in {a,b,c}{
		\draw (\i11) edge (\i12);
		\draw (\i12) edge (\i13);
		\draw (z) edge (\i11);}
		
	\end{tikzpicture}
	\caption{Topology graph in the proof of \Cref{thm:cyclicfs:treetop}. \label{fig:topoex:treetop}}
\end{figure}

First, observe that $W$ must be a connected set of nodes. 
Indeed, if there is a proper subset of agents $C\subsetneq N$ that occupy a connected component of the subgraph of $G$ induced by $W$, then there exists an agent $i\in C$ such that $i+1\not\in C$. 
Moreover, there is an empty node on the (unique) path in $G$ connecting $\asgnm(i)$ with $\asgnm(i+1)$, and therefore $i$ has an incentive to jump to this node.

The observation above implies that $z\in W$, and for every $x\in \{a,b,c\}$, there exists $r_x\in \{0,1,2,3\}$ such that $\{x_1,x_2,x_3\}\cap W = \{x_i\colon i\in [r_x]\}$. 
In particular, among the $a$-type nodes in \Cref{fig:topoex:treetop}, the set of occupied nodes is $\emptyset$, $\{a_1\}$, $\{a_1,a_2\}$, or $\{a_1,a_2,a_3\}$; analogous statements hold for the $b$- and $c$-type nodes.
Assume without loss of generality that $r_a\ge r_b\ge r_c$. 
We perform a case analysis based on the vector $(r_a, r_b, r_c)$.
\begin{itemize}
\item $(r_a,r_b,r_c) = (3,2,0)$.
Consider the agent $i = \asgnm^{-1}(b_2)$. 
The agent that has $i$ as a friend must be placed at $b_1$; otherwise, she would benefit by jumping to $b_3$. 
But since $c_1$ is closer to every node except $b_1$ than $b_2$ is,
$i$ can benefit by jumping to~$c_1$, a contradiction.

\item $(r_a,r_b,r_c) = (3,1,1)$. 
The unique agent that has $\asgnm^{-1}(c_1)$ as a friend must be placed at $z$; otherwise, she would benefit by jumping to $c_2$. 
However, the same must hold for the unique agent that has $\asgnm^{-1}(b_1)$ as a friend, which is impossible.

\item $(r_a,r_b,r_c) = (2,2,1)$.
Assume without loss of generality that $\asgnm(1) = b_2$. 
Then, $\asgnm(6) = b_1$, as otherwise agent~$6$ would have a beneficial jump to $b_3$. 
This in turn implies $\asgnm(2) = z$, because otherwise agent~$1$ can jump to either $a_3$ or $c_2$ to improve her utility. 
Let $s\in \{3,4,5\}$ be such that $\asgnm(s) = a_2$. 
The same chain of arguments as for agent~$1$ yields that $\asgnm(s+1) = z$, which is impossible since $\asgnm(2) = z$.
\end{itemize}
We have reached a contradiction in all cases, so $\asgnm$ cannot be jump stable.
\end{proof}

Interestingly, if the topology graph is an ``extended star'' with three branches and the friendship graph is a cycle, as in the instance used in the proof of \Cref{thm:cyclicfs:treetop}, then a jump stable assignment always exists whenever there are at least $16$ agents.
In fact, we show a more general result with an arbitrary number of branches.
To this end, we define an \emph{extended star} as a tree in which only one vertex, called the \emph{center}, has degree at least~$3$. 
A \emph{branch} of an extended star is a path of maximal length such that one endpoint is the center.
The \emph{size} of a branch is the number of nodes on the branch, not counting the center of the star.

\begin{theorem}
\label{thm:extended-star}
Let $k\ge 3$ be an integer, and assume that $G$ is an extended star with $k$ branches.
Suppose that there are $n\ge 5k+1$ agents with non-negative utilities, and the friendship graph is a cycle.
For any distance factor function, there exists a jump stable assignment.
\end{theorem}

\begin{proof}
Assume that the agents are ordered $1,2,\dots,n$ according to the cycle in the friendship graph.
Call each branch of size at most $4$ a \emph{type-1 branch} and each branch of size at least $5$ a \emph{type-2 branch}.
We assign the agents following their order branch by branch, starting with type-1 branches, then type-2 branches, and finally the center node.
For each type-1 branch, we fill the entire branch, whereas for each type-2 branch, we assign at least $5$ agents to the branch.
Since there are at least $5k+1$ agents and only $k$ branches, we have enough agents for type-2 branches.

For each type-1 branch, we fill in the agents from the leaf towards the center.
For each type-2 branch, let $r$ be the number of agents that we want to assign to the branch; we may choose this number arbitrarily as long as it is at least $5$. 
We divide into two cases depending on the parity of $r$.
\begin{itemize}
\item If $r = 2s$ is even, we assign the first $s$ agents starting from the node closest to the center and leaving one empty node after assigning each agent (except the $s$-th agent).
We then fill in the other $s$ agents starting from the subsequent node and moving back towards the center.
\item If $r = 2s+1$ is odd, we proceed similarly except that we start with the \emph{second} closest node to the center.
\end{itemize}
Finally, we assign the last agent to the center node.
An example with $k = 3$ branches of size $4,6,7$ is shown in \Cref{fig:extended-star}.

\begin{figure*}
    \centering	\begin{tikzpicture}[
	element/.style={shape=circle,draw, fill=white,scale=0.7}
	]
	
		\node[element] (z) at (0,0) {\color{white}$66$};
		\node at (0,0) {$16$};
		\node[element] (a1) at (200:1.3)  {\color{white}$66$};
		\node at (200:1.3) {$4$};
		\foreach[count = \l] \i/\j in {1.3/3,2.6/2,3.9/1}{
        \pgfmathsetmacro{\k}{int(\l + 1)}
		\node[element] (a\k) at ($(a1) + (200:0.8*\i)$) {\color{white}$66$};
		\node at ($(a1) + (200:0.8*\i)$) {$\j$};
		}
		\node[element] (b1) at (160:1.3)  {\color{white}$66$};
		\node at (160:1.3) {$9$};
		\foreach[count = \l] \i/\j in {1.3/5,2.6/8,3.9/6,5.2/7,6.5/{}}{
        \pgfmathsetmacro{\k}{int(\l + 1)}
		\node[element] (b\k) at ($(b1) + (160:0.8*\i)$) {\color{white}$66$};
		\node at ($(b1) + (160:0.8*\i)$) {$\j$};
		}
		\foreach[count = \k] \i/\j in {1.3/10,2.6/15,3.9/11,5.2/14,6.5/12,7.8/13,9.1/{}}{
		\node[element] (c\k) at (0:0.8*\i) {\color{white}$66$};
		\node at (0:0.8*\i) {$\j$};
		}
		\foreach[count = \i] \j in {2,3,4}{
		\draw (a\i) edge (a\j);}
		\foreach[count = \i] \j in {2,3,4,5,6}{
		\draw (b\i) edge (b\j);}
		\foreach[count = \i] \j in {2,3,4,5,6,7}{
		\draw (c\i) edge (c\j);}
		\foreach \i in {a,b,c}
		\draw (z) edge (\i1);
		
	\end{tikzpicture}
	\caption{Topology graph in the proof of \Cref{thm:extended-star} when $n = 16$, $k = 3$, and the three branches have size $4$, $6$, and $7$. \label{fig:extended-star}}
\end{figure*}
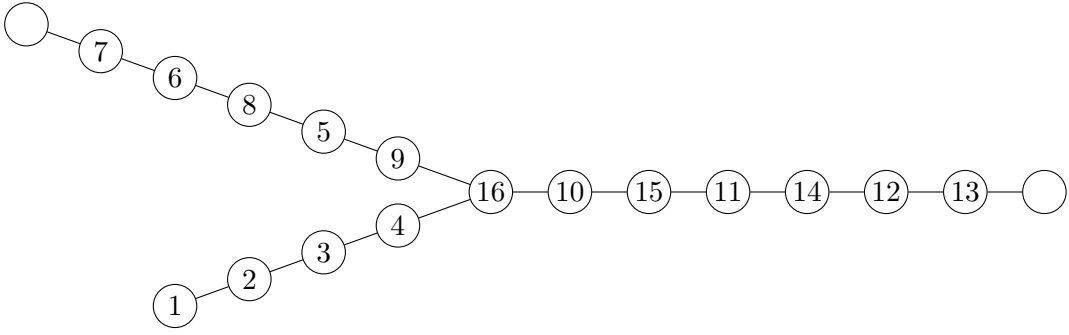

Since $n \ge 5k+1$, there is at least one type-2 branch.
We now show that the resulting assignment is jump stable.
\begin{itemize}
\item Consider agent~$n$ assigned to the center node.
If there is a type-1 branch, agent~$n$ cannot get closer to her friend, agent~$1$, because every type-1 branch is completely filled.
Otherwise, agent~$n$ is at distance at most~$2$ from agent~$1$, and cannot get closer to agent~$1$ because agent~$1$'s branch contains at least five agents.
\item For each type-1 branch, every agent is already next to her friend except the agent next to the center node.
For this latter agent, her friend is in either a type-1 branch---in which case she cannot get closer because the branch is already filled---or a type-2 branch---in which case she is at distance at most~$3$ from her friend and every empty node has distance at least~$4$ from this friend.
\item Consider a type-2 branch with an even number of agents $2s$.
The $s$-th agent (in order of agent number) is already next to her friend.
The $(2s)$-th agent is at distance at most~$4$ from her friend and every empty node has distance at least~$4$ from this friend.
Every other agent on this branch is at distance~$2$ from her friend, and no node adjacent to this friend is empty.
A similar argument applies when the branch contains an odd number of agents.
\end{itemize}
Hence, in all cases, there is no beneficial jump.
\end{proof}

Since jump stable assignments are not guaranteed to exist in TDGs, a natural question is whether we can efficiently decide if such an assignment exists.
As it is known that determining whether a \emph{Nash stable partition} exists in ASHGs is \NP-hard \citep{SuDi10a}, the connection between ASHGs and TDGs that we outlined in \Cref{sec:related} implies that the jump stability question for TDGs is also \NP-hard.\footnote{The non-existence of stable states and computational hardness results regarding stability are also known for jump stability in Schelling games \citep{AEGI21a,KBFN22a}. Since these games have additional features such as agent types and fractional utilities, which are not shared by TDGs, these results do not carry over to our setting.}
However, using reductions between the two classes of games results in TDG instances where the number of connected components is linear in $n$, the number of nodes is quadratic in $n$, and the distance factor function does not play any role (because every connected component is a clique).
We provide a novel reduction to show that the hardness still holds even for instances that better reflect the essence of TDGs.

\begin{restatable}{theorem}{exjump}
\label{thm:exjump}
For the reciprocal distance factor function, it is \NP-complete to decide whether there exists a jump stable assignment in a given \tdg.
This holds even if the topology graph~$G$ consists of a constant number of connected components and $|V| = \Theta(n)$.
\end{restatable}

At a high level, for the hardness, we provide a reduction from the \NP-complete problem \textsc{Exact $3$-Cover} \citep{Karp72a}. 
An instance of \textsc{Exact $3$-Cover} consists of a tuple $(R,S)$, where~$R$ is a ground set whose size is a multiple of $3$ and $S$ is a collection of $3$-element subsets of~$R$. 
An instance is a Yes-instance if there exists a subcollection $S'\subseteq S$ of size $|R|/3$ that exactly covers~$R$.

Given an instance $(R,S)$ of \textsc{Exact $3$-Cover}, we construct a TDG instance as follows. 
There are three types of agents.
The first two types represent the elements of $R$ and the sets in $S$. 
The last type, which consists of only one agent, is a ``disturber'' strictly avoided by agents representing sets in $S$. 
The topology consists of four connected components. 
One component serves to hold all agents representing~$R$ and agents representing an exact $3$-cover from $S$. 
The edges of this component are designed in such a way that it is possible to assign agents to all nodes of this component if and only if the original instance is a Yes-instance. 
The other three components are cliques of carefully chosen sizes which either allow all other agents representing sets in~$S$ to flee the disturber (in case of a Yes-instance), or provide sufficient space to allow for a run-and-chase dynamics that prevents stability (in case of a No-instance).

\begin{proof}[Proof of \Cref{thm:exjump}]

First, note that membership in \NP\ is clear, because given an assignment~$\asgnm$, one can check if it is jump stable in polynomial time by simply checking for every agent~$i$ and every empty node~$v$ whether $u_i(\asgnm^{i\to v}) > u_i(\asgnm)$.

For the hardness, we provide a reduction from the \NP-complete problem \textsc{Exact $3$-Cover}, whose description can be found before this proof. 
Note that \textsc{Exact $3$-Cover} remains \NP-complete even if we assume that for each instance $(R,S)$, it holds that 
\begin{equation}
\label{eq:R-and-S}
|R| > 3 \quad\text{ and }\quad 3|S| > |R| > 3|S| - \frac{|R|}{3} + 10.    
\end{equation}
Indeed, for the first lower bound on $|R|$, instances with $|R| \le 3$ are easy. 
Also, if $3|S| < |R|$, then the instance is trivially a No-instance, whereas if $3|S| = R$, we can simply check whether each element in $R$ occurs exactly once in~$S$.
For the second lower bound on $|R|$, we can pad every instance $(R,S)$ by adding $3k$ new elements $\{z_i^j\colon j\in [3],i\in [k]\}$ to~$R$ and $k$ new sets $\{\{z_i^j\colon j\in [3]\} \colon i\in [k]\}$ to~$S$. 
This leads to a new instance $(R',S')$ with $|R'| = |R| + 3k$ elements and $|S'| = |S| + k$ sets. 
Note that $|R'| > 3|S'| - |R'|/3 + 10$ if and only if $|R| + 3k > 3(|S|+k) - (|R| + 3k)/3 + 10$, which is equivalent to $k > 3|S| -4|R|/3 + 10$. 
Hence, this padding can be done with a polynomial blowup of the original instance, and we obtain a reduced instance which satisfies the desired inequalities and which is a Yes-instance if and only if it originates from a Yes-instance.

Let $(R,S)$ be an instance of \textsc{Exact $3$-Cover} satisfying $|R| > 3$ and $3|S| > |R| > 3|S| - |R|/3 + 10$. 
To construct a \tdg, we let $N = N_R \cup N_S \cup \{c\}$, where $N_R = \{a_r\colon r\in R\}$ and $N_S = \{b_s\colon s\in S\}$, and let $\pi\colon R \to [|R|]$ be an arbitrary bijection. 
Then, we define the utilities as follows.
\begin{itemize}
    \item For all $r, r'\in R$, let 
    \[
    u_{a_r}(a_{r'}) = 
    \begin{cases}
       1 & \text{if } \pi(r') - \pi(r) \equiv 1\;(\bmod\; |R|);\\
       0 & \textnormal{otherwise.}
    \end{cases}
    \]
    \item For all $i\in N_R$ and $j\in N\setminus N_R$, let $u_i(j) = 0$.
    \item For all $s\in S$ and $r\in R$, let 
    \[
    u_{b_s}(a_{r}) = 
    \begin{cases}
       1 & \text{if } r\in s;\\
       -\frac{6}{|R|-3} & \textnormal{otherwise.}
    \end{cases}
    \]
    \item For all $s\in S$, let $u_{b_s}(j) = 0$ for every $j\in N_S$, and let $u_{b_s}(c) = - 10$.
    \item Let
    \[
    u_{c}(j) = 
    \begin{cases}
       1 & \text{if } j\in N_S;\\
       0 & \text{if } j\in N_R.
    \end{cases}
    \]
\end{itemize}
Note that $u_{b_s}(a_r)$ is well-defined because we assume that $|R| > 3$. 

It remains to specify the topology graph. 
This graph consists of four connected components, three of which are cliques of specific sizes, while the final component is a clique with further attached nodes that will mimic exact covers.
The fourth component is illustrated in \Cref{fig:V4}.
Formally, let $G = (V,E)$, where $V = \bigcup_{i\in [4]}V_i$ with $V_i = \{v_i^j\colon j\in [|S| +2]\}$ for $i\in [2]$, $V_3 = \{v_3^j\colon j\in [|S|-|R|/3]\}$, and $V_4 = \{v_4^j\colon j\in [|R| + |R|/3]\}$. 
The edge set is defined as follows.
\begin{itemize}
    \item For $i\in [3]$ and $v,w\in V_i$ with $v\ne w$, $\{v,w\}\in E$.
    \item For $i,j\in [|R|]$ with $i\ne j$, $\{v_4^i, v_4^j\}\in E$.
    \item For $i\in [|R|/3]$, it holds that $\{v_4^{3i-2},v_4^{|R|+i}\}, \{v_4^{3i-1},v_4^{|R|+i}\}, \{v_4^{3i},v_4^{|R|+i}\}\in E$.
    \item No further edges are in $E$.
\end{itemize}
This completes the description of our \tdg.
\begin{figure}
    \centering
    \begin{tikzpicture}[
	element/.style={shape=circle,draw, fill=white, inner sep = 0pt, minimum size = 25pt, scale = 0.9}
	]
		\draw (0,0) circle (3cm);
		\node at (0,0) {\Large $K_{12}$};
		\foreach[count = \k] \i in {5,4,3,2,1,0,11,10,9,8,7,6}
		{\node[element] (v\k) at (15+30*\i:2.5) {$v_4^{\k}$};}
		\node[element] (w1) at (135:4) {$v_4^{13}$};
		\node[element] (w4) at (225:4) {$v_4^{16}$};
		\node[element] (w3) at (315:4) {$v_4^{15}$};
		\node[element] (w2) at (45:4) {$v_4^{14}$};
		
        \foreach \j in {1,2,3,4}{
		\pgfmathparse{int(3*\j)}
		\draw (v\pgfmathresult) edge (w\j);
		\pgfmathparse{int(3*\j-1)}
		\draw (v\pgfmathresult) edge (w\j);
		\pgfmathparse{int(3*\j-2)}
		\draw (v\pgfmathresult) edge (w\j);
        } 
    \end{tikzpicture}
    \caption{Fourth component of the topology graph in the reduced instance in the proof of \Cref{thm:exjump} for the case $|R| = 12$.
	The component consists of a clique $K_{|R|}$ with $|R|$ nodes; in the figure, we omit the edges of the clique.
	Additionally, there are $|R|/3$ nodes that are each connected to a triplet of nodes from the clique.}
    \label{fig:V4}
\end{figure}
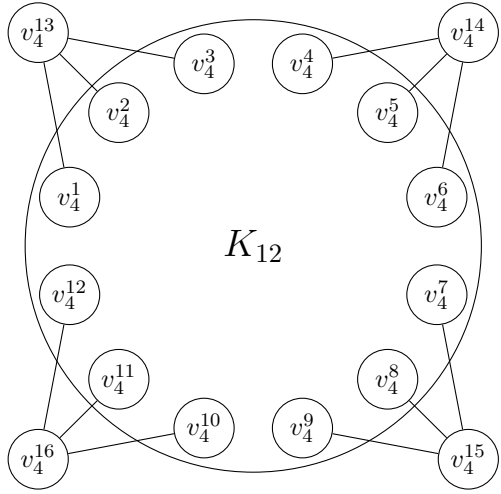

We now show that there exists a jump stable assignment in our \tdg if and only if $(R,S)$ is a Yes-instance.
First, assume that $(R,S)$ is a Yes-instance, and let $S'\subseteq S$ be an exact cover of $R$. 
Let $\sigma\colon R\to [|R|]$ be a bijection such that $\{\sigma^{-1}(3i-2),\sigma^{-1}(3i-1),\sigma^{-1}(3i)\}\in S'$ for all $i\in [|R|/3]$. 
In other words, $\sigma$ labels the elements in $R$ in such a way that the elements belonging to the same subset in $S'$ have consecutive numbers. 
Also, let $\tau\colon S\setminus S'\to [|S|-|R|/3]$ be an arbitrary bijection. 
Now, consider the assignment $\asgnm\colon N\to V$ with
\begin{itemize}
    \item $\asgnm(a_{\sigma^{-1}(i)}) = v_4^i$ for all $i\in [|R|]$,
    \item $\asgnm(b_s) = v_4^{|R|+i}$ for $s\in S'$ such that $s = \{\sigma^{-1}(3i-2),\sigma^{-1}(3i-1),\sigma^{-1}(3i)\}$,
    \item $\asgnm(b_s) = v_3^{\tau(s)}$ for $s\in S\setminus S'$, and
    \item $\asgnm(c) = v_1^1$.
\end{itemize}

We claim that $\asgnm$ is jump stable. Let $v\in V$ be an arbitrary empty node, i.e., $v\in (V_1\cup V_2)\setminus \{v_1^1\}$.
\begin{itemize}
    \item For $r\in R$, it holds that $u_{a_r}(\asgnm^{a_r\to v}) = 0 < u_{a_r}(\asgnm)$.
    \item For $s\in S$, it holds that $u_{b_s}(\asgnm^{b_s\to v}) \le 0,$
    whereas $u_{b_s}(\asgnm) = 3\cdot 1 + (|R|-3)\cdot\frac{1}{2}\cdot\left(-\frac{6}{|R|-3}\right) = 0$; here, the factor $\frac12$ is due to the reciprocal distance factor function. 
    \item It holds that $u_c(\asgnm^{c\to v}) = 0 = u_c(\asgnm)$. 
\end{itemize}
Hence, $\asgnm$ is jump stable.

Conversely, assume that there exists a jump stable assignment $\asgnm$. We make a key observation for the proof in the next claim.

\begin{claim}\label{cl:VRinV4}
In a jump stable assignment $\asgnm$, it holds that $\asgnm(N_R)\subseteq V_4$. 
\end{claim}
\renewcommand\qedsymbol{$\vartriangleleft$} 
\begin{proof}
    We first present a roadmap for the proof of this claim.
    We start by proving that if an agent in $N_R$ occupies a node in a component $V_i$, then all nodes in $V_i$ must be occupied.
    We continue by showing that such an agent must exist in $V_4$ and no such agent can exist in $V_1$ or $V_2$.
    Consequently, the components $V_3$ and $V_4$ are fully occupied, and there is only one agent remaining.
    We then show that $c$ cannot be in $V_4$, which leads us to the fact that an agent from $N_S$ occupies a node in the clique of $V_4$.
    Arguments regarding the utility of this agent allow us to establish the claim by contradiction.

    We now prove the claim formally.
    Assume for contradiction that there exists an agent $i\in N_R$ with $\asgnm(i)\notin V_4$. 
    Then, there does not exist any $i\in [4]$ such that $\asgnm(N_R)\subseteq V_i$, because the sets~$V_i$ have size smaller than $|R|$ for $i\in [3]$, due to (\ref{eq:R-and-S}). 
    We will show that for every $i\in [4]$, if $V_i\cap \asgnm(N_R) \neq \emptyset$, then it holds that $V_i\subseteq \asgnm(N)$, i.e., all nodes in $V_i$ are occupied. 
    
    Consider therefore $i\in [4]$ with $V_i\cap \asgnm(N_R) \neq \emptyset$. 
    Assume for contradiction that there exists $v\in V_i\setminus \asgnm(N)$. 
    Because the agents in $N_R$ are assigned to at least two connected components of $G$, there exists $j\in [|R|]$ such that $a_{\pi^{-1}(j)}\in \asgnm^{-1}(V_i)$ and $a_{\pi^{-1}(j-1)}\notin \asgnm^{-1}(V_i)$. 
    (For convenience, if $j = 1$ we consider $j-1$ to be $|R|$.) 
    Hence, $u_{a_{\pi^{-1}(j-1)}}(\asgnm^{a_{\pi^{-1}(j-1)}\to v}) > 0 = u_{a_{\pi^{-1}(j-1)}}(\asgnm)$. 
    This is a contradiction to the jump stability of $\asgnm$.
    It must therefore be the case that $V_i\subseteq \asgnm(N)$.
    
    Next, note that $|V_1| + |V_2| + |V_3| = 3|S|-|R|/3 + 4$. 
    By (\ref{eq:R-and-S}), it holds that $|N_R| > |V_1| + |V_2| + |V_3|$. 
    Hence, there exists an agent $x\in N_R$ with $\asgnm(x)\in V_4$. 
    Our consideration in the previous paragraph implies that $V_4\subseteq \asgnm(N)$.
    Note that $|V_4| + |V_i| > |N|$ for $i\in [2]$, so it cannot be the case that $V_i\subseteq \asgnm(N)$ for some $i\in [2]$.
    Thus, $V_i \cap \asgnm(N_R) = \emptyset$ for $i\in [2]$. 
    This means that $\asgnm(N_R)\cap V_3 \ne\emptyset$, and consequently $V_3\subseteq \asgnm(N)$. 
    Now, $|V_3| + |V_4| = |N| - 1$. 
    Therefore, $|\asgnm^{-1}(V_1\cup V_2)| = 1$, and it must hold for every agent $x\in N_S$ that $u_x(\asgnm) \ge 0$, because $x$ could jump to $V_1$ or $V_2$ (whichever component is empty) to obtain non-negative utility.
    
    We show that there exists $s\in S$ with $\asgnm(b_s) \in \{v_4^i\colon i\in [|R|]\}$, i.e., an agent corresponding to a set in $S$ occupies a position in the large clique of $V_4$.
    First, $|V_3| + 1 < |S|$ because $|R| >3$.
    Therefore, there exists $s'\in S$ with $\asgnm(b_{s'}) \in V_4$.
    This implies that $\asgnm(c)\notin V_4$; indeed, otherwise we would have 
    \[
    u_{b_{s'}}(\asgnm) \le \frac 13 u_{b_{s'}}(c) + \sum_{r\in {s'}}u_{b_{s'}}(a_r) = -\frac{10}{3} + 3 < 0,
    \]
    contradicting the non-negative utility requirement of $b_{s'}$. 
    The existence of a desired $s\in S$ follows because $|\asgnm^{-1}(V_4) \cap N_R|\le |R|-1$ and $V_4\subseteq \asgnm(N)$.
    
    We derive a final contradiction by estimating the utility of an agent~$b_s$ such that $\asgnm(b_s) \in \{v_4^i\colon i\in [|R|]\}$ whose existence is guaranteed by the previous paragraph. 
    To this end, observe that by (\ref{eq:R-and-S}), we have $|R| \ge 3|S| - |R|/3 + 11$, which means that $4|R|/3 \ge 3|S| + 11$.
    Dividing both sides by $3$, we obtain $4|R|/9 \ge |S| + 11/3$, or equivalently, $|R| - |S| - 1 \ge 5|R|/9 + 8/3$. 
    Hence, at least $5|R|/9 + 8/3$ agents in $N_R$ are assigned to a node in $\{v_4^i\colon i\in [|R|]\}$ adjacent to the node occupied by~$b_s$. 
    Recall that $b_s$ has non-positive utility for other agents in $N_S$ as well as for $c$.
    Letting $T = \asgnm^{-1}(\{v_4^i\colon i\in [|R|]\})\setminus \{a_r\colon r\in s\}$, it follows that 
    \begin{align*}
u_{b_s}(\asgnm) 
&\le \sum_{r\in s} u_{b_s}(a_r) + \sum_{t\in T} u_{b_s}(t)\\
& \le 3 - \left(\left(\frac{5|R|}{9} + \frac83\right) - 3\right) \cdot \frac{6}{|R| - 3}
< 3 - \frac {10|R|-6}{3|R| - 9} < 0\text.
    \end{align*}
    This is a contradiction to the required non-negative utility of $b_s$, and our final contradiction for this claim.
\end{proof}

Observe that, for $i\in [2]$, it holds that $|V_i| = |N_S| + 2 > |N_S| + 1 = |N| - |N_R|$.
Combining this with \Cref{cl:VRinV4}, we find that there must be at least one empty node in each of $V_1$ and $V_2$.
If $u_{b_s}(\asgnm) < 0$ for some $s\in S$, then $b_s$ has a beneficial jump to at least one of these empty nodes.
Hence, $u_{b_s}(\asgnm) \ge 0$ for all $s\in S$.
In order for this to be true, $c$ cannot be assigned to the same connected component as any agent in~$N_S$.
This means that $u_c(\asgnm) = 0$.
To avoid a beneficial jump for~$c$, no agent from~$N_S$ can be in either $V_1$ or $V_2$.
Since $|V_3| + |V_4| = |N_R| + |N_S|$, it follows that $V_3\cup V_4$ is filled with exactly the agents from $N_R\cup N_S$.

Define $S' := \{s\in S\colon b_s\in \asgnm^{-1}(V_4)\cap N_S\}$. 
We will show that $S'$ partitions~$R$. 
Note that $|S'| = |V_4| - |N_R| = |R|/3$. 
Therefore, it suffices to show that all sets in $S'$ are pairwise disjoint. For any $s\in S'$, by the consideration in the last paragraph of the proof of \Cref{cl:VRinV4}, $\asgnm(b_s)\notin \{v_4^i\colon i\in [|R|]\}$.  
Hence, $\asgnm(\{b_s\colon s\in S'\}) = \{v_4^i\colon |R| + 1\le i \le 4|R|/3\}$, which means that $\asgnm(N_R) = \{v_4^i\colon 1\le i \le |R|\}$.
Let $s\in S'$, and suppose that $\asgnm(b_s) = v_4^{|R| + i}$ for some $i\in [|R|/3]$. 
We claim that $\asgnm(\{a_r\colon r\in s\}) = \{v_4^{3i-2},v_4^{3i-1},v_4^{3i}\}$. 
Indeed, if this is not the case, i.e., if some of the agents corresponding to elements of $R$ in $s$ are not assigned as a neighbor of $b_s$ in~$G$, then
\begin{align*}
    u_{b_s}(\asgnm) &= \sum_{r\in R} \frac{1}{\mathrm{d}_G(\asgnm(b_s),\asgnm(a_r))}\cdot u_{b_s}(a_r)\\
    &\le 2 + \frac 12 - \frac{6}{|R|-3} - \frac 12 (|R|-4)\cdot \frac{6}{|R|-3}\\
    & <  2 + 1 - \frac{3}{|R|-3}  - \frac 12 (|R|-4)\cdot \frac{6}{|R|-3}
    = 3 - 3\cdot \frac{|R|-3}{|R|-3} = 0\text,
\end{align*}
contradicting the required non-negative utility of agent~$b_s$.
Hence, we can indeed conclude that $\asgnm(\{a_r\colon r\in s\}) = \{v_4^{3i-2},v_4^{3i-1},v_4^{3i}\}$. Since this argument holds for any set $s\in S'$, we have that the sets in $S'$ are all disjoint.
This means that $S'$ partitions $R$, as desired.
\renewcommand\qedsymbol{$\square$}
\end{proof}

\section{Dynamics}
\label{sec:dynamics}

In this section, we investigate the dynamics induced by performing beneficial jumps. 
Dynamics offer an interesting distributed perspective on stability and have been examined, for instance, in hedonic games \citep{FMM21a,BBT22a,BBW21a,BBK23a}. 
We are interested in the following questions:
\begin{itemize}
    \item Given an initial assignment, is it \emph{possible} that the dynamics converges, i.e., does there exist a sequence of beneficial jumps that results in a jump stable assignment?
    \item Given an initial assignment, is it \emph{necessary} that the dynamics converges, i.e., do all sequences of beneficial jumps result in a jump stable assignment?
\end{itemize}

In the hedonic games literature, the complexity of these questions was first studied (in the conference version of the work) by \citet{BBW21a} in several classes of hedonic games \emph{excluding} ASHGs. 
Subsequently, \citet{BBT22a} presented conditions for the necessary convergence of dynamics in ASHGs. 
However, since previous work did not consider computational boundaries for dynamics in ASHGs, we cannot transfer such results to our setting.

We now analyze these questions for TDGs.
First, we observe that dynamics are guaranteed to converge for symmetric utilities, because the potential function in the proof of \Cref{thm:symmetric-jump-existence} increases with every beneficial jump.
In this sense, the proof yields more than merely the existence of a jump stable assignment. 
For asymmetric utilities, since a jump stable assignment may not exist, convergence of the dynamics is also no longer guaranteed. 
Nevertheless, if the friendship graph is  acyclic and utilities are non-negative, the convergence is retained. 

\begin{theorem}
\label{thm:dynamic-acyclic}
    For any distance factor function, if the friendship graph is acyclic and utilities are non-negative, then the jump dynamics is guaranteed to converge.
\end{theorem}

\begin{proof}
Consider an instance of \tdg such that the friendship graph is acyclic and utilities are non-negative. As in the proof of \Cref{thm:DAGassignment}, there exists a topological order $\pi$ of the agents in $N$, i.e., a function $\pi\colon N\to [n]$ such that $\pi(i) > \pi(j)$ whenever $u_i(j) > 0$. 
We shall define a potential function based on this order.

Consider a sequence of assignments $(\asgnm_k)_{k\ge 1}$, where for each $k\ge 1$, $\asgnm_{k+1} = \asgnm_k^{d_k\to v_k}$ for some agent $d_k$ and node $v_k$ which is empty in $\asgnm_k$. 
Associate any assignment $\asgnm$ with the vector 
\begin{align*}
    \Lambda(\asgnm) = (u_{\pi^{-1}(1)}(\asgnm), u_{\pi^{-1}(2)}(\asgnm),\dots, u_{\pi^{-1}(n)}(\asgnm))\text.
\end{align*}
For two vectors $\mathbf{x} = (x_1,\dots,x_n)$ and $\mathbf{y} = (y_1,\dots,y_n)$, we say that $\mathbf{x}$ \emph{lexicographically dominates} $\mathbf{y}$, denoted by $\mathbf{x} >_{\mathrm{lex}} \mathbf{y}$, if there exists $i\in \{1,\dots,n\}$ such that $x_j = y_j$ for all $j\in \{1,\dots,i-1\}$ and $x_i > y_i$.

We claim that for every $k\ge 1$, it holds that $\Lambda(\asgnm_{k+1}) >_{\mathrm{lex}} \Lambda(\asgnm_k)$.
Let $k\ge 1$ and consider the deviator $d_k$.
By definition of the topological order, for each $j < \pi(d_k)$, we have $u_{\pi^{-1}(j)}(d_k) = 0$.
Hence, for every $j < \pi(d_k)$, it holds that $u_{\pi^{-1}(j)}(\asgnm_{k+1}) = u_{\pi^{-1}(j)}(\asgnm_k)$, because this utility is not affected by $d_k$'s jump. 
It follows that the first $\pi(d_k) - 1$ entries of $\Lambda(\asgnm_k)$ and $\Lambda(\asgnm_{k+1})$ are identical.
Moreover, as $d_k$ improves her utility with the beneficial jump, the $\pi(d_k)$-th entry of $\asgnm_k$ increases, and therefore $\Lambda(\asgnm_{k+1}) >_{\mathrm{lex}} \Lambda(\asgnm_k)$.
\end{proof}

\Cref{thm:dynamic-acyclic} also shows that there exists a potential function with respect to which the jump dynamics is increasing. 
This implies that the problem of finding jump stable outcomes in instances with an acyclic friendship graph and non-negative utilities is contained in the complexity class \PLS. 
However, it is unlikely that this problem is \PLS-complete, because we can compute jump stable assignments for this class of games in polynomial time (cf.~\Cref{thm:DAGassignment}). 
In fact, \PLS-completeness for this problem would imply that $\PLS = \P$. 
Still, similarly to \PLS-complete problems, the jump dynamics, which is the natural local search dynamics for our problem, can run in exponential time.

\begin{theorem}\label{thm:expdirdyn}
    For any distance factor function, the jump dynamics may run for an exponential number of steps, even if the friendship graph is acyclic and utilities are non-negative.
\end{theorem}

\begin{proof}
    We prove that, for each $k\ge 1$, there exists an additively separable hedonic game~$H_k$ with $n = 2k + 1$ agents, non-negative utilities, and an acyclic friendship graph such that the jump dynamics runs for at least $2^k$ steps.
    Recall from \Cref{sec:related} that ASHGs form a subclass of TDGs.
    
    Let $k\ge 1$ and define an ASHG $(N,u)$, where $N = \{a_j,b_j\colon j\in [k]\}\cup \{a_0\}$ is a set of $2k+1$ agents and the utilities are given as follows:
    \begin{itemize}
        \item For each $j\in [k]$,  $u_{a_{j}}(b_j) = 1$ and $u_{a_{j}}(a_{j-1}) = 2$. 
        \item All other utilities are set to $0$. 
    \end{itemize}
    Clearly, the utilities are non-negative and the friendship graph is acyclic. An illustration of the friendship graph is given in \Cref{fig:expdirdyn}.
    
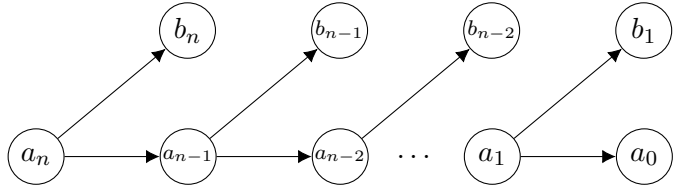
\begin{figure}
    \centering
    \begin{tikzpicture}[
	element/.style={shape=circle,draw, fill=white, scale = 0.9}
	]
        \node[element] (a0) at (0,0) {\color{white}$66$};
        \node[element] (b1) at ($(a0) + (40:2.6)$) {\color{white}$66$};
        \node[element] (a1) at ($(a0) + (0:2)$) {\color{white}$66$};
        \node[element] (b2) at ($(a1) + (40:2.6)$) {\color{white}$66$};
        \node[element] (a2) at ($(a1) + (0:2)$) {\color{white}$66$};
        \node[element] (b3) at ($(a2) + (40:2.6)$) {\color{white}$66$};
        \node (dots) at ($(a2) + (0:1)$) {$\dots$};
        \node[element] (a3) at ($(dots) + (0:1)$) {\color{white}$66$};
        \node[element] (b4) at ($(a3) + (40:2.6)$) {\color{white}$66$};
        \node[element] (a4) at ($(a3) + (0:2)$) {\color{white}$66$};
        
        \node at (a0) {$a_n$};
        \node at (b1) {$b_n$};
        \node at (a1) {\scriptsize $a_{n-1}$};
        \node at (b2) {\scriptsize $b_{n-1}$};
        \node at (a2) {\scriptsize $a_{n-2}$};
        \node at (b3) {\scriptsize $b_{n-2}$};

        \node at (a3) {$a_1$};
        \node at (b4) {$b_1$};
        \node at (a4) {$a_0$};
        
        \foreach \j in {1,2,4}{
		\pgfmathparse{int(\j-1)}
		\draw[->] (a\pgfmathresult) edge (a\j);
		\draw[->] (a\pgfmathresult) edge (b\j);
        } 
		\draw[->] (a2) edge (b3);
    \end{tikzpicture}
    \caption{Friendship graph of the game in the proof of \Cref{thm:expdirdyn}.}
    \label{fig:expdirdyn}
\end{figure}
We show by induction on $k$ that there exists a sequence of beneficial jumps in $H_k$ such that agent $a_k$ performs at least $2^k$ jumps. 
For the base case $k = 1$, start with a partition of the agents into singletons, and let $a_1$ first join $b_1$ and then~$a_0$. 

Assume that we have constructed the dynamics for some $k\ge 1$. 
Observe that $H_{k+1}$ is an extension of $H_k$ with the addition of agents $a_{k+1}$ and $b_{k+1}$. 
Given the dynamics constructed for $H_k$, we extend it to one for $H_{k+1}$ as follows.
Agents $a_{k+1}$ and $b_{k+1}$ are initially in singleton coalitions. 
After every jump by agent $a_k$ in the original dynamics, we insert the following jumps: $a_{k+1}$ joins $b_{k+1}$, then $a_{k+1}$ joins the coalition containing~$a_k$. 
This leads to a total of $2\cdot 2^k = 2^{k+1}$ jumps by $a_{k+1}$. 
Notice that all jumps in this extended sequence (whether by $a_{k+1}$ or not) are beneficial jumps. 
Indeed, original jumps are not affected by the new agents $a_{k+1}$ and $b_{k+1}$, because these two agents do not influence the utilities of original agents. 
Moreover, whenever $a_k$ jumps, she always leaves the coalition containing $a_{k+1}$.
The only exception is $a_k$'s first jump, when $a_{k+1}$ is still in a singleton coalition; however, even for this jump, $a_k$ does not join $a_{k+1}$.
Hence, agent $a_{k+1}$ has a utility of $0$ after each jump by $a_k$, so all jumps by $a_{k+1}$ are beneficial jumps.
\end{proof}

Since we make use of the connection between ASHGs and TDGs in \Cref{thm:expdirdyn}, the topology graph of the constructed game has a non-constant number of connected components. 
We conjecture that exponential running time is possible even when the number of components is constant. 

Next, we show that the questions of whether the jump dynamics starting from an initial assignment possibly or necessarily converges are both computationally hard.

\begin{restatable}{theorem}{dynconv}\label{thm:dynconv}
For any distance factor function, deciding whether the jump dynamics possibly converges is \NP-hard, even if utilities are restricted to be non-negative.
\end{restatable}

As in the proof of \Cref{thm:exjump}, we reduce from  \textsc{Exact $3$-Cover}. 
Given an instance $(R,S)$ of \textsc{Exact $3$-Cover}, we construct a TDG containing, among other agents, three agents whose friendship graph forms a cycle. 
The jump dynamics will run into a cycle due to these three agents; the only way to prevent this is through a jump by a special agent $\delta$ initially assigned to a connected component containing agents representing the sets in~$S$.
The agent $\delta$ can disrupt the cycle if and only if agents representing an exact $3$-cover of $R$ jump to a connected component containing agents representing the elements of~$R$.

\begin{proof}[Proof of \Cref{thm:dynconv}]
Let an arbitrary distance factor function $f\colon \mathbb{Z}_{\ge 1}\to \mathbb{R}_{> 0}$ be given. We provide a reduction from \textsc{Exact $3$-Cover}.
Let $(R,S)$ be an instance of \textsc{Exact $3$-Cover}. 
We may assume without loss of generality\footnote{This can be assured by a similar padding technique as in the proof of \Cref{thm:exjump}. We add $3k$ new elements to $R$ and every $3$-element subset of these new elements to $S$, for a sufficiently large~$k$. 
This ensures the desired condition with a polynomial-size blowup of the instance.} that $|S|> \frac 23 |R|$.
To construct a \tdg, we let $N = N_A \cup N_R \cup N_S$, where $N_A = \{\alpha_1,\alpha_2,\alpha_3,\delta, \rho_1,\rho_2\}\cup\{\gamma_i\colon i\in [|R|]\}$, $N_R = \{a_r\colon r\in R\}$, and $N_S = \{b_s\colon s\in S\}$. The sets $N_R$ and $N_S$ consist of agents representing elements in the ground set $R$ and subsets in $S$, respectively. The agents in~$N_A$ are auxiliary agents with the following roles:
\begin{itemize}
    \item Agents $\alpha_1,\alpha_2,\alpha_3$ induce a cyclic friendship graph and can therefore potentially generate a cycle in the jump dynamics.
    \item Agent $\delta$ is the only agent that can interrupt this cycle.
    \item Agents in $N_R$ have as their only objective to be as close as possible to $\rho_1$. 
    Agent $\rho_2$ ensures that their distance to $\rho_1$ is initially~$2$.
    \item Agents $\gamma_i$ for $i\in [|R|]$ can block certain positions.
\end{itemize}
Let $\epsilon \in \left(0, \frac {f(1)-f(2)}{2f(1) + f(2)}\right)$; this is possible because $f$ is strictly decreasing. 
We will define utilities based on $\epsilon$. 
Note that the reduced instance is of polynomial size, because $\epsilon$ only depends on the distance factor function. 
We formally define the utilities as follows.
\begin{itemize}
    \item Let $u_{\alpha_1}(\alpha_2) = u_{\alpha_2}(\alpha_3) = u_{\alpha_3}(\alpha_1) = 1$.
    \item Let $u_{\delta}(\alpha_1) =|S| - \frac {|R|}3 + \frac 12$. For $s\in S$, let $u_{\delta}(b_s) = 1$.
    \item For $s\in S$, let $u_{b_s}(\delta) = 1$. Moreover, for $r\in s$, let $u_{b_s}(a_r) = \frac 13 (1 + \epsilon)$.
    \item For $r\in R$, let $u_{a_r}(\rho_1) = 1$.
    \item For $i\in [|R|]$, let $u_{\gamma_i}(\rho_2) = 1$.
    \item For all pairs of agents $i,j\in N$ such that $u_i(j)$ has not been specified, let $u_i(j) = 0$. 
    In particular, this ensures that agents $\rho_1$ and $\rho_2$ are indifferent between all assignments, and never have an incentive to jump.
\end{itemize}

Next, we define the topology graph, which consists of four connected components, and the initial assignment.
The first component of the topology graph is illustrated in \Cref{fig:dyncycle}, the second and fourth components are complete graphs, and the third component is a cycle of length $4$.
The first component contains all agents from $N_R$ as well as $\rho_1$ and $\rho_2$. 
The second component contains all agents from $N_S$ together with $\delta$.
The role of the third component is to ensure cyclic behavior of $\alpha_1,\alpha_2,\alpha_3$ unless agent~$\delta$ disrupts this behavior.
The fourth component is a clique that ``stores'' agents $\{\gamma_i\colon i\in [|R|]\}$ before they jump to the first component.

Formally, let $G = (V,E)$ with $V = \bigcup_{i\in [4]}V_i$, where 
\begin{itemize}
    \item $V_1 = \{z_1,z_2\}\cup \{x_i\colon i\in [|R|]\}\cup \{y_i^j\colon i\in [|R|/3], j\in [4]\}$,
    \item $V_2 = \{w_i\colon i \in [|S|+1]\}$, 
    \item $V_3 = \{t_i\colon i\in [4]\}$, and
    \item $V_4 = \{q_i\colon i\in [|R|]\}$.
\end{itemize}
The edge set is defined as follows.
\begin{itemize}
    \item For $i \in [|R|]$, $\{x_i,z_2\}\in E$.
    \item For $i\in [|R|/3]$ and $j,k\in [4]$ with $j\neq k$, $\{y_i^j,y_i^k\}\in E$.
    \item For $i\in [|R|/3]$ and $j\in [3]$, $\{y_i^j,z_1\}\in E$.
    \item $\{z_1,z_2\}\in E$.
    \item For $i,j \in [|S|+1]$ with $i\neq j$, $\{w_i,w_j\}\in E$.
    \item $\{t_1,t_2\}, \{t_2,t_3\}, \{t_3,t_4\}, \{t_4,t_1\}\in E$.
    \item For $i,j \in [|R|]$ with $i\neq j$, $\{q_i,q_j\}\in E$.
    \item No further edges are in $E$.
\end{itemize}

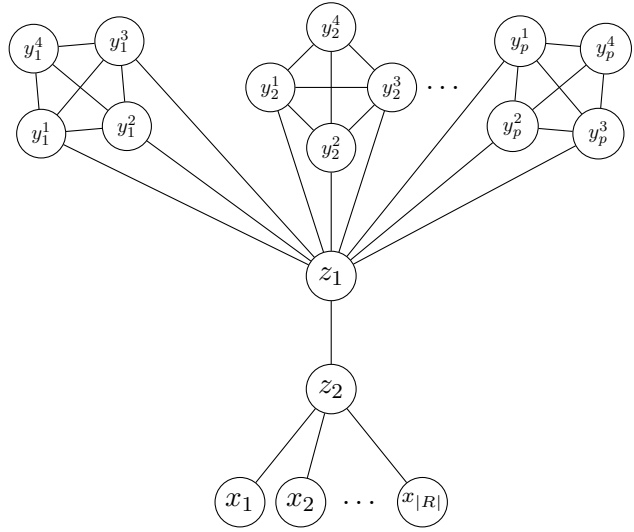
\begin{figure}
    \centering
\begin{tikzpicture}[
	element/.style={shape=circle,draw, fill=white,scale=0.7}
	]
    \node[element] (z1) at (0,0) {\color{white}\scriptsize $x_{|R|}$};
    \node[element] (z2) at (0,-1.5) {\color{white}\scriptsize $x_{|R|}$};
    \node at (0,0) {$z_1$};
    \node at (0,-1.5) {$z_2$};
    \node[element] (r1) at ($(z2) + (-1.2, -1.5)$) {\color{white}\scriptsize $x_{|R|}$};
    \node[element] (r2) at ($(z2) + (-.4, -1.5)$) {\color{white}\scriptsize $x_{|R|}$};
    \node at ($(z2) + (-1.2, -1.5)$) {$x_1$};
    \node at ($(z2) + (-.4, -1.5)$) {$x_2$};
    \node[element] (r3) at ($(z2) + (1.2, -1.5)$) {\color{white}\scriptsize $x_{|R|}$};
    \node at (r3) {\scriptsize $x_{|R|}$};
    
    \node at ($(z2) + (.4, -1.5)$) {$\dots$};
    
    \foreach[count = \k] \i/\j/\s in {-3.3/50/1,0/0/2, 3/310/p}{
    \node (yc\k) at ($(z1) + (\i, 2.5)$) {};
    \foreach[count = \l] \r in {180,270,0,90}{
    \pgfmathsetmacro{\rot}{int(\j + \r)}
    \node[element] (y\k\l) at ($(yc\k) + (\rot:.8)$) {$y_\s^\l$};
    }
    \foreach \s in {1,2,3} {
    \draw (z1) edge (y\k\s);
    }
    \draw (y\k1) edge (y\k2);
    \draw (y\k1) edge (y\k3);
    \draw (y\k1) edge (y\k4);
    \draw (y\k2) edge (y\k3);
    \draw (y\k2) edge (y\k4);
    \draw (y\k3) edge (y\k4);
    
    } 
    \node at ($(z1) + (1.5, 2.5)$) {$\dots$};

    \draw (z1) edge (z2);
    \foreach \i in {1,2,3}{
    \draw (r\i) edge (z2);

    }

\end{tikzpicture}
    \caption{Illustration of the first component of the topology graph of the reduced instance in \Cref{thm:dynconv,thm:dyncycle}, where $p = |R|/3$.}
    \label{fig:dyncycle}
\end{figure}

Let arbitrary permutations $\pi \colon R \to [|R|]$ and $\tau\colon S\to [|S|]$ be given. 
We specify the initial assignment $\asgnm\colon N \to V$ as follows.
\begin{itemize}
    \item For $r\in R$, $\asgnm(a_r) = x_{\pi(r)}$.
    \item For $i\in [2]$, $\asgnm(\rho_i) = z_i$.
    \item For $s\in S$, $\asgnm(b_s) = w_{\tau(s)}$.
    \item $\asgnm(\delta) = w_{|S|+1}$.
    \item For $i\in [3]$, $\asgnm(\alpha_i) = t_i$.
    \item For $i\in [|R|]$, $\asgnm(\gamma_i) = q_i$.
\end{itemize}

We claim that the jump dynamics starting from assignment~$\asgnm$ can converge in the reduced \tdg if and only if the \textsc{Exact $3$-Cover} instance $(R,S)$ is a Yes-instance.

First, assume that $(R,S)$ is a Yes-instance, i.e., there exists a subcollection $S'\subseteq S$ that partitions $R$.
Let $\phi: S' \to [|R|/3]$ be an arbitrary bijection.
The following sequence of jumps leads to a stable assignment:
\begin{itemize}
    \item For each $s\in S'$, let the agents in the set $\{a_r\colon r\in s\}$ jump to $\{y_{\phi(s)}^i\colon i\in [3]\}$. 
    These are beneficial jumps because these agents increase their utility from $f(2)$ to $f(1)$.
    \item For each $s\in S'$, let $b_s$ jump to $y_{\phi(s)}^4$. This is beneficial because $b_s$ improves her utility from $f(1)$ to $3\cdot\frac 13 (1+\epsilon)f(1) = (1+\epsilon)f(1)$.
    \item Let $\delta$ jump to $t_4$. This is an improvement for $\delta$ from a utility of $\left(|S| - \frac{|R|}3\right)\cdot f(1)$ to $\left(|S| - \frac{|R|}3 + \frac 12 \right)\cdot f(1)$.
    \item Let $\{\gamma_i \colon i\in [|R|]\}$ jump to $\{x_1,\dots, x_{|R|}\}$.
\end{itemize}

We claim that we have reached a state where no agent has a beneficial jump. 
To see this, note that all nodes in $V_1\cup V_3$ are occupied, and therefore no jump to either of these components is possible.
Also, all nodes in $V_4$ are empty, so since all agents already receive non-negative utility, they do not have a beneficial jump to $V_4$.
Finally, $V_2$ only contains agents of the type $b_s$, which means that only $\delta$ may have an incentive to jump there, but $\delta$ is better off staying at $t_4$.
Consequently, the jump dynamics can converge if $(R,S)$ is a Yes-instance.

Conversely, assume that the jump dynamics can converge in the reduced \tdg when starting from assignment $\asgnm$.

Note that the instance restricted to agents $\alpha_1,\alpha_2,\alpha_3$ and topology restricted to the subgraph induced by $V_3$ is an instance where no jump stable assignment exists (\Cref{thm:cycle-in-cycle}). 
Since agents $\alpha_1,\alpha_2,\alpha_3$ cannot improve their utility by jumping to a node outside $V_3$, the dynamics will cycle indefinitely unless some agent jumps to a node in~$V_3$, and the only agent who has a beneficial jump of this type is $\delta$. 
By jumping to $V_3$, agent~$\delta$ receives a utility of at most $\left(|S| - \frac{|R|}3 + \frac 12 \right)\cdot f(1)$. 
Apart from a jump to $V_3$, agent~$\delta$ could only jump when following agents of the type $b_s$. 
Since $|S| > \frac 23 |R|$, this would only be possible if more than $|R|/3$ such agents have already jumped. 
As a result,
$\delta$ can only jump once at least $|R|/3$ agents in $N_S$ have left $V_2$. 

Let $S'$ consist of the first $|R|/3$ agents $b_s$ to jump in the dynamics.
Recall that $\rho_1$ and $\rho_2$ are indifferent between all assignments and never have an incentive to jump.
Hence, agents $a_r$ for $r\in R$ would never leave the first component.
This means that the agents in $S'$ would only jump to this component when they jump for the first time.
Each such agent $b_s$ would only jump there if she has the agents $\{a_r\colon r\in s\}$ as direct neighbors, because otherwise she would receive utility at most 
\begin{align*}
    2 \cdot\frac 13 (1 + \epsilon) f(1) + \frac 13 (1 + \epsilon) f(2) 
    &= \frac 13 (1+\epsilon)(2f(1) + f(2))\\
    & < \frac 13 \left(1+\frac {f(1)-f(2)}{2f(1) + f(2)}\right)(2f(1) + f(2)) = f(1)\text,
\end{align*}
which would not be an improvement for her.

Since agents $\rho_1$ and $\rho_2$ never leave $z_1$ and $z_2$, and agents $a_r$ for $r\in R$ do not have an incentive to jump to $y_i^4$ for $i\in [|R|/3]$, it must be that for each $s\in S'$, the three agents in $\{a_r\colon r\in s\}$ jump to the nodes $\{y_i^j\colon j\in [3]\}$ for some $i\in [|R|/3]$, and then $b_s$ jumps to $y_i^4$. 
In particular, this implies that the sets in $S'$ are disjoint, and so $S'$ partitions $R$.
Hence, $(R,S)$ is a Yes-instance, as desired.
\end{proof}

The proof for deciding whether the dynamics necessarily converges is similar. 
In this case, the special agent~$\delta$ initially blocks cycling and has to perform a jump to initiate cycling. 
In contrast to the reduction in the proof of \Cref{thm:dynconv}, we have an additional set of auxiliary agents of type $\sigma_i$, which are liked by $\delta$ and which initially occupy nodes in the neighborhood of $\delta$.
The idea is that $\delta$ can only perform a jump after all agents $\sigma_i$ have left their initial nodes, and then $\delta$ wants to follow them.
To make this possible, a set of agents representing an exact $3$-cover of the original instance has to vacate their nodes to create space for the agents $\sigma_i$. 
Apart from some technical differences, large parts of the proof employ similar arguments as the proof of \Cref{thm:dynconv}.
In particular, the first component of the topology graph is similar, and the desired jump dynamics for the agents representing $R$ and $S$ is the same.

\begin{restatable}{theorem}{dyncycle}\label{thm:dyncycle}
For any distance factor function, deciding whether the jump dynamics necessarily converges is \coNP-hard, even if utilities are restricted to be non-negative.
\end{restatable}

\begin{proof}

It suffices to show that deciding whether the jump dynamics possibly \emph{cycles} is \NP-hard.

Let an arbitrary distance factor function $f\colon \mathbb{Z}_{\ge 1}\to \mathbb{R}_{> 0}$ be given. We provide a reduction from \textsc{Exact $3$-Cover}.
Let $(R,S)$ be an instance of \textsc{Exact $3$-Cover}.
To construct a \tdg, we let $N = N_A \cup N_R \cup N_S \cup \{c\}$, where $N_A = \{\alpha_1,\alpha_2,\alpha_3,\delta, \rho_1,\rho_2, \sigma\}\cup\{\sigma_i\colon i \in[|R|/3]\}$, $N_R = \{a_r\colon r\in R\}$, and $N_S = \{b_s\colon s\in S\}$. The sets $N_R$ and $N_S$ consist of agents representing elements in the ground set $R$ and subsets in $S$, respectively. The agents in $N_A$ are auxiliary agents with the following roles:
\begin{itemize}
    \item Agents $\alpha_1,\alpha_2,\alpha_3$ are the only agents that can generate cycling of the dynamics.
    \item Agent $\delta$ has to jump in order to start the cycling.
    \item Agent $\sigma$ is in the connected component of the agents in $N_S$ and attracts agents $\sigma_i$ for $i \in [|R|/3]$. The only possible deviation of agent $\delta$ is to eventually join this component.
    \item Agents in $N_R$ have as their only objective to be as close as possible to $\rho_1$.
    Agent~$\rho_2$ ensures that their distance to $\rho_1$ is initially~$2$.
\end{itemize}
Let $\epsilon \in \left(0, \frac {f(1)-f(2)}{2f(1) + f(2)}\right)$; this is possible because $f$ is strictly decreasing. 
We will define utilities based on $\epsilon$. 
Note that the reduced instance is of polynomial size, because $\epsilon$ only depends on the distance factor function. 
We formally define the utilities as follows.
\begin{itemize}
    \item Let $u_{\alpha_1}(\alpha_2) = u_{\alpha_2}(\alpha_3) = u_{\alpha_3}(\alpha_1) = 1$.
    \item For $i \in[|R|/3]$, let $u_{\delta}(\sigma_i) = 1$.
    Let $u_{\delta}(\alpha_1) =\frac {|R|}3 - \frac 12$. 
    \item For $i \in[|R|/3]$, let $u_{\sigma_i}(\sigma) = 1$.
    \item For $s\in S$, let $u_{b_s}(\sigma) = 1$. Moreover, for $r\in s$, let $u_{b_s}(a_r) = \frac 13 (1 + \epsilon)$.
    \item For $r\in R$, let $u_{a_r}(\rho_1) = 1$.
    \item For all pairs of agents $i,j\in N$ such that $u_i(j)$ has not been specified, let $u_i(j) = 0$. 
    In particular, this ensures that agents $\rho_1, \rho_2, \sigma$ are indifferent between all assignments, and never have an incentive to jump.
\end{itemize}

Next, we define the topology graph, which consists of three connected components, and the initial assignment.
The first component of the topology graph is identical to the first component in the proof of \Cref{thm:dynconv} and is illustrated in \Cref{fig:dyncycle}, the second component is a complete graph, while the third component is a star with $|R|/3$ leaves whose center is also part of a cycle of length $4$.
The first component contains all agents from $N_R$ as well as $\rho_1$ and $\rho_2$.
The second component contains all agents from $N_S$ together with $\sigma$.
The role of the third component is to initiate cyclic behavior if the TDG is derived from a Yes-instance of \textsc{Exact $3$-Cover}.

Formally, let $G = (V,E)$ with $V = \bigcup_{i\in [3]}V_i$, where 
\begin{itemize}
    \item $V_1 = \{z_1,z_2\}\cup \{x_i\colon i\in [|R|]\}\cup \{y_i^j\colon i\in [|R|/3], j\in [4]\}$,
    \item $V_2 = \{w_i\colon i \in [|S|+2]\}$, and 
    \item $V_3 = \{t_i\colon i\in [4]\}\cup \{v_i\colon i\in [|R|/3]\}$.
\end{itemize}
The edge set is defined as follows.
\begin{itemize}
    \item For $i \in [|R|]$, $\{x_i,z_2\}\in E$.
    \item For $i\in [|R|/3]$ and $j,k\in [4]$ with $j\neq k$, $\{y_i^j,y_i^k\}\in E$.
    \item For $i\in [|R|/3]$ and $j\in [3]$, $\{y_i^j,z_1\}\in E$.
    \item $\{z_1,z_2\}\in E$.
    \item For $i,j \in [|S|+2]$ with $i\neq j$, $\{w_i,w_j\}\in E$.
    \item For $i\in [|R|/3]$, $\{t_4,v_i\}\in E$.
    \item $\{t_1,t_2\}, \{t_2,t_3\}, \{t_3,t_4\}, \{t_4,t_1\}\in E$.
    \item No further edges are in $E$.
\end{itemize}

Let arbitrary permutations $\pi \colon R \to [|R|]$ and $\tau\colon S\to [|S|]$ be given. 
We specify the initial assignment $\asgnm\colon N \to V$ as follows.
\begin{itemize}
    \item For $r\in R$, $\asgnm(a_r) = x_{\pi(r)}$.
    \item For $i\in [2]$, $\asgnm(\rho_i) = z_i$.
    \item For $s\in S$, $\asgnm(b_s) = w_{\tau(s)}$.
    \item $\asgnm(\sigma) = w_{|S|+1}$.
    \item For $i\in [3]$, $\asgnm(\alpha_i) = t_i$.
    \item For $i\in [|R|/3]$, $\asgnm(\sigma_i) = v_i$.
    \item $\asgnm(\delta) = t_4$.
\end{itemize}

We claim that the jump dynamics starting at assignment~$\asgnm$ can cycle in the reduced \tdg if and only if the \textsc{Exact $3$-Cover} instance $(R,S)$ is a Yes-instance.

First, assume that $(R,S)$ is a Yes-instance, i.e., there exists a subcollection $S'\subseteq S$ that partitions $R$.
Let $\phi: S' \to [|R|/3]$ be an arbitrary bijection.
The following sequence of jumps leads to cycling:
\begin{itemize}
    \item For each $s\in S'$, let the agents in the set $\{a_r\colon r\in s\}$ jump to $\{y_{\phi(s)}^i\colon i\in [3]\}$. 
    These are beneficial jumps because these agents increase their utility from $f(2)$ to $f(1)$.
    \item For each $s\in S'$, let $b_s$ jump to $y_{\phi(s)}^4$. 
    This is beneficial because $b_s$ improves her utility from $f(1)$ to $3\cdot\frac 13 (1+\epsilon)f(1) = (1+\epsilon)f(1)$.
    \item For each $i\in [|R|/3]$, let $\sigma_i$ jump to an empty node left by an agent $b_s$ for $s\in S'$; this is possible since $|S'| = |R|/3$.
    Each agent $\sigma_i$'s utility increases from $0$ to $f(1)$. 
    \item Let $\delta$ jump to $w_{|S|+2}$. This is an improvement for $\delta$ from a utility of $\left(\frac{|R|}3 - \frac 12\right)\cdot f(1)$ to $\frac{|R|}3\cdot f(1)$.
    \item The agents $\{\alpha_i\colon i\in [3]\}$ can perform an infinite cycle of beneficial jumps on the $4$-cycle induced by $\{t_i\colon i\in [4]\}$, by \Cref{thm:cycle-in-cycle}.
\end{itemize}

Hence, the jump dynamics can cycle if $(R,S)$ is a Yes-instance.

Conversely, assume that the jump dynamics can cycle in the reduced \tdg when starting from assignment $\asgnm$.
Recall that agents $\rho_1,\rho_2,\sigma$ will never jump.
As a result, agents $a_r$ for $r\in R$ will at most jump to a direct neighbor of $\rho_1$.

Let $S' = \{s\in S\colon b_s \text{ jumps in the dynamics}\}$.
For $s\in S$, agent~$b_s$ may only jump to the first connected component. 
Each such agent would only jump there if she has the agents $\{a_r\colon r\in s\}$ as direct neighbors, because otherwise she would receive utility at most 
\begin{align*}
    2\cdot \frac 13 (1 + \epsilon) f(1) + \frac 13 (1 + \epsilon) f(2) 
    &= \frac 13 (1+\epsilon)(2f(1) + f(2))\\
    & < \frac 13 \left(1+\frac {f(1)-f(2)}{2f(1) + f(2)}\right)(2f(1) + f(2))
    = f(1)\text,
\end{align*}
which would not be an improvement for her. 
Therefore, the only possibility for $b_s$ to jump is that the three agents in $\{a_r\colon r\in s\}$ jump to the nodes $\{y_i^j\colon j\in [3]\}$ for some $i\in [|R|/3]$, and then $b_s$ jumps to $y_i^4$. 
In particular, this implies that the sets in $S'$ are disjoint.
Note also that once $b_s$ is at $y_i^4$, she will not make any more jumps.

We will conclude the proof by showing that $|S'|\ge |R|/3$, which implies that the sets in $S'$ cover all elements in $R$.
To see this, observe that agents $\sigma_i$ for $i\in [|R|/3]$ can jump at most once, and this jump must be to the second component of $G$.
Moreover, agents $\alpha_i$ for $i\in [3]$ can only jump after agent $\delta$ does. 
Since we have already seen that each agent outside $\{\alpha_i\colon i\in [3]\}\cup \{\delta\}$ can jump at most once, the dynamics can cycle only if agent~$\delta$ jumps.

Now, the only agents relevant for $\delta$'s jump are $\alpha_1$ and $\sigma_i$ for $i \in [|R|/3]$.
Hence, the only possible jump of $\delta$ is to the second component. 
An inspection of the utilities of $\delta$ implies that this is possible only if all agents in $\{\sigma_i \colon i\in [|R|/3]\}$ have already jumped there \emph{and} there is an empty position for $\delta$ to jump to. 
Hence, at least $|R|/3$ nodes in the second component have to be vacated during the dynamics. 
The agents who vacate these nodes must correspond to the agents associated with the sets in $S'$. 
Hence, $|S'| \ge |R|/3$.
Since the sets in $S'$ are disjoint, it must be that $|S'| = |R|/3$, and therefore $(R,S)$ is a Yes-instance.
\end{proof}

\section{Discussion}

In this work, we have introduced the model of topological distance games (TDGs), which aim to capture scenarios in which the utility of an agent depends on both her inherent utilities for other agents and her distance from them.
We presented results on the existence, computational, and dynamical properties of jump stable assignments in our model.
While such assignments may not exist even under weak assumptions, it turns out that existence guarantees can be obtained for symmetric utilities as well as in the presence of structured friendship relations.

Given that TDGs combine important aspects of other models in coalition formation, including hedonic games, social distance games, and Schelling games (see \Cref{sec:related}), our study may inspire further work from several angles.
For instance, while we have shown that a jump stable assignment exists in a number of cases, one could try to obtain a more complete characterization of the topology and friendship graphs that admit such an assignment or even guarantee convergence of the jump dynamics.
Understanding precisely when it is possible to efficiently determine whether a jump stable assignment exists is also an interesting direction.
Yet another potential avenue is to extend our results to \emph{weighted} graphs, wherein the distance between two agents is the length of the shortest path (in terms of the sum of weights) between their assigned nodes. 
Finally, in addition to jump stability, other notions such as \emph{swap stability} or \emph{envy-freeness} are worth exploring in TDGs as well---we provide some initial results on swap stability in \Cref{app:swap}.

\section*{Acknowledgments}

This work was partially supported by the Deutsche Forschungsgemeinschaft under grants BR 2312/11-2 and BR 2312/12-1, by the Singapore Ministry of Education under grant
number MOE-T2EP20221-0001, and by an NUS Start-up Grant.
We thank the anonymous reviewers of the 37th AAAI Conference on Artificial Intelligence (AAAI 2023) and Theoretical Computer Science for their valuable feedback.

\bibliographystyle{plainnat}
\bibliography{aaai23}

\appendix

\section{Swap Stability}
\label{app:swap}

In this section, we provide some results on the notion of swap stability.
Given an assignment $\asgnm$ and agents $i,j\in N$, denote by $\asgnm^{i\leftrightarrow j}$ the assignment that results when $i$ and $j$ swap positions.
Note that with swap stability, we do not need to assume that $|V| > n$; in particular, we can allow $|V| = n$ as well.

\begin{definition}
Given an instance of TDG, an assignment~$\asgnm$ is said to be \emph{swap stable} if for every pair of agents $i,j\in N$, $u_i(\asgnm^{i\leftrightarrow j})\le u_i(\asgnm)$ or $u_j(\asgnm^{i\leftrightarrow j}) \le u_j(\asgnm)$.
\end{definition}

If utilities are symmetric, we can derive analogous results for swaps as we did in \Cref{sec:symmetric} for jumps.

\begin{proposition}
\label{prop:symmetric-swap-existence}
For any distance factor function and symmetric utilities, there exists a swap stable assignment.
\end{proposition}

\begin{proof}
Consider an assignment maximizing the potential function $\Phi(\asgnm) := \sum_{i\in N}u_i(\asgnm)$; such an assignment must exist because the number of possible assignments is finite.
Assume for contradiction that agents $k$ and $k'$ both prefer to swap with each other.
This means that 
\begin{align}
u_{k}(\asgnm) < u_{k}(\asgnm^{k\leftrightarrow k'}) \text{ and } u_{k'}(\asgnm) < u_{k'}(\asgnm^{k\leftrightarrow k'}). 
\label{eq:potential-swap}
\end{align}
Let $v = \asgnm(k)$ and $w = \asgnm(k')$.
When $k$ and $k'$ swap positions, the potential function changes by 
\begin{align*}
 \Phi(\asgnm^{k\leftrightarrow k'}) - \Phi(\asgnm) &= \sum_{i\in N}u_i(\asgnm^{k\leftrightarrow k'}) - \sum_{i\in N}u_i(\asgnm) \\
&= (u_{k}(\asgnm^{k\leftrightarrow k'}) - u_k(\asgnm)) + (u_{k'}(\asgnm^{k\leftrightarrow k'}) - u_{k'}(\asgnm)) \\
&\quad + \sum_{j\in N\setminus\{k,k'\}} [f(d_G(\asgnm(j), w)) - f(d_G(\asgnm(j), v))]\cdot u_j(k) \\
&\quad + \sum_{j\in N\setminus\{k,k'\}} [f(d_G(\asgnm(j), v)) - f(d_G(\asgnm(j), w))]\cdot u_j(k') \\
&= (u_{k}(\asgnm^{k\leftrightarrow k'}) - u_k(\asgnm)) + (u_{k'}(\asgnm^{k\leftrightarrow k'}) - u_{k'}(\asgnm)) \\
&\quad + \sum_{j\in N\setminus\{k,k'\}} [f(d_G(w,\asgnm(j))) - f(d_G(v,\asgnm(j)))]\cdot u_k(j) \\
&\quad + \sum_{j\in N\setminus\{k,k'\}} [f(d_G(v,\asgnm(j))) - f(d_G(w,\asgnm(j)))]\cdot u_{k'}(j) \\
&= (u_{k}(\asgnm^{k\leftrightarrow k'}) - u_k(\asgnm)) + (u_{k'}(\asgnm^{k\leftrightarrow k'}) - u_{k'}(\asgnm)) \\
&\quad + \sum_{j\in N\setminus\{k,k'\}} [f(d_G(w,\asgnm(j))) - f(d_G(v,\asgnm(j)))]\cdot u_k(j) \\
&\quad + [f(d_G(w,v)) - f(d_G(v,w))]\cdot u_k(k') \\
&\quad + \sum_{j\in N\setminus\{k,k'\}} [f(d_G(v,\asgnm(j))) - f(d_G(w,\asgnm(j)))]\cdot u_{k'}(j) \\
&\quad + [f(d_G(v,w)) - f(d_G(w,v))]\cdot u_{k'}(k) \\
&= 2\left(u_{k}(\asgnm^{k\leftrightarrow k'}) - u_k(\asgnm)\right) + 2\left(u_{k'}(\asgnm^{k\leftrightarrow k'}) - u_{k'}(\asgnm) \right) 
> 0,
\end{align*}
where we use the symmetry of the utilities for the second equality and \eqref{eq:potential-swap} for the inequality.
This contradicts the assumption that $\asgnm$ maximizes the potential function $\Phi$.
\end{proof}

\begin{proposition}
\label{prop:symmetric-swap-PLS}
For any distance factor function and symmetric utilities, finding a swap stable assignment is \PLS-complete.
\end{proposition}

\begin{proof}
Membership in \PLS\ is clear, as we can consider the potential $\Phi(\asgnm) := \sum_{i\in N}u_i(\asgnm)$ from the proof of \Cref{prop:symmetric-swap-existence} as the measure to be maximized. 
Checking for a local improvement only requires the inspection of a polynomial number of swaps.

For \PLS-hardness, we provide a reduction from the \PLS-complete problem \textsc{Graph Partitioning} under the \textsc{Swap}-neighborhood \citep{ScYa91a}.
An instance of \textsc{Graph Partitioning} consists of a complete and undirected weighted graph on $2t$ vertices. 
A solution is a partition of the vertex set into two subsets of size $t$, where the \emph{cut} between the two subsets, i.e., the total weight of edges between vertices of the two subsets, is to be (locally) minimized.
Given a solution of an instance of \textsc{Graph Partitioning}, its \textsc{Swap}-neighborhood contains all partitions in which exactly one vertex from each subset is moved to the other subset, i.e., the two vertices are swapped.

Consider an instance $H = (V_P,E_P,w)$ of \textsc{Graph Partitioning} where $V_P$, $E_P$, and $w\colon E_P\to \mathbb R$ are the vertex set, the edge set, and the weight function, respectively.
Let $t := |V_P|/2$.
We define the reduced \tdg as follows. 
Let $N = \{\alpha_v\colon v\in V_P\}$ be the set of agents. 
The utility functions are given by $u_{\alpha_x}(\alpha_y) = w(\{x,y\})$, where $x,y\in V_P$; clearly, this defines a symmetric utility function.
The topology graph $G = (V,E)$ is given by $V = A\cup B$ where $A = \{a_i\colon i\in [t]\}$, $B = \{b_i\colon i\in [t]\}$, and $E = \{\{a_i,a_j\},\{b_i,b_j\}\colon 1\le i < j\le t\}$. 
In other words, the topology graph consists of two cliques of size $t$.

Now, every assignment $\asgnm$ induces the $2$-partition $P_{\asgnm} = (A_P, B_P)$ with $A_P = \{x\in V_P\colon \alpha_x\in \asgnm^{-1}(A)\}$ and $B_P = \{x\in V_P\colon \alpha_x \in \asgnm^{-1}(B)\}$.
Note that each of the two subsets in the partition has size $t$.

Consider a pair of vertices $x\in A_P$ and $y\in B_P$. 
The change in the weight of the cut when $x$ and $y$ swap their partition classes, that is, the change in weight when going from the cut induced by $P_{\asgnm}$ to the cut induced by $P_{\asgnm^{\alpha_x\leftrightarrow \alpha_y}}$, is exactly 
\begin{align*}
     &\Bigg(\sum_{z\in A_P\setminus \{x\}} w(x,z) - \sum_{z\in B_P} w(x,z)\Bigg) 
     + \Bigg(\sum_{z\in B_P\setminus \{y\}} w(y,z) - \sum_{z\in A_P} w(y,z)\Bigg) \\
     &= \Bigg(\sum_{z\in A_P\setminus \{x\}} w(x,z) - \sum_{z\in B_P\setminus\{y\}} w(x,z)\Bigg) 
     + \Bigg(\sum_{z\in B_P\setminus \{y\}} w(y,z) - \sum_{z\in A_P\setminus\{x\}} w(y,z)\Bigg) \\
    &=  (u_{x}(\asgnm) -  u_x(\asgnm^{x\leftrightarrow y})) + (u_{y}(\asgnm) -  u_y(\asgnm^{x\leftrightarrow y}))\text.
\end{align*}

Hence, by the computations in \Cref{prop:symmetric-swap-existence}, the change in the value of a cut after $x\in A_P$ and $y\in B_P$ swap their partition classes in $P_{\asgnm}$ is exactly half of the negative of the change of the potential $\Phi$ when going from $\asgnm$ to $\asgnm^{x\leftrightarrow y}$. 
Therefore, $P_{\asgnm}$ is locally optimal if and only if $\asgnm$ is. Note that the change in the sign of the differences is important because \textsc{Graph Partitioning} is a minimization problem whereas finding a swap stable assignment is a maximization problem (with respect to the potential $\Phi$).
\end{proof}

As is the case with jump stability, a swap stable assignment is no longer guaranteed to exist for asymmetric utilities, and deciding the existence of such an assignment is \NP-complete.
In fact, for swap stability, these results already follow from prior work on the \emph{roommates problem}.

\begin{proposition}
\label{prop:roommate-nonexistence}
For any distance factor function, there exists an instance of TDG with no swap stable assignment.
This holds even if the topology graph~$G$ is a union of disjoint edges and utilities are non-negative.
\end{proposition}

\begin{proof}
The proposition follows from a roommates problem instance due to \citet{Alca94a}; for convenience, we reproduce the proof in our terminology.
Let $G$ consist of two disjoint edges and $n = 4$, and assume that the four agents have the utilities shown in the table below.
\renewcommand{\arraystretch}{1.2}
\begin{center}
\begin{tabular}{|c||c|c|c|c|}
\hline
Agent  & $1$ & $2$ & $3$ & $4$ \\
\hline \hline
$u_1$ & $0$ & $3$ & $2$ & $1$ \\
\hline
$u_2$ & $2$ & $0$ & $1$ & $3$ \\
\hline
$u_3$ & $3$ & $1$ & $0$ & $2$ \\
\hline
$u_4$ & $1$ & $2$ & $3$ & $0$ \\
\hline
\end{tabular}
\end{center}
We claim that this instance does not admit a swap stable assignment.
Indeed, if agents~$1$ and $2$ are together, then agents~$2$ and $3$ prefer to swap.
If agents~$1$ and $3$ are together, then agents~$1$ and $4$ prefer to swap.
Finally, if agents~$1$ and $4$ are together, agents~$1$ and~$2$ prefer to swap.
\end{proof}

\begin{proposition}
\label{prop:roommate-NP}
For any distance factor function, it is \NP-complete to decide whether there exists a swap stable assignment in a given TDG.
This holds even if the topology graph $G$ is a union of disjoint edges and utilities are non-negative.
\end{proposition}

\begin{proof}
The problem belongs to $\NP$ since we can verify that an assignment is swap stable in polynomial time by checking all pairs of agents.
The \NP-completeness is an immediate consequence of the fact that the exchange-stable roommates problem is \NP-complete \citep[Cor.~3.2]{CeMa05a}.
\end{proof}

From the dynamical point of view, we provide an example where a swap stable assignment exists, but \emph{any} sequence of beneficial swaps starting from a different assignment inevitably runs into a cycle.

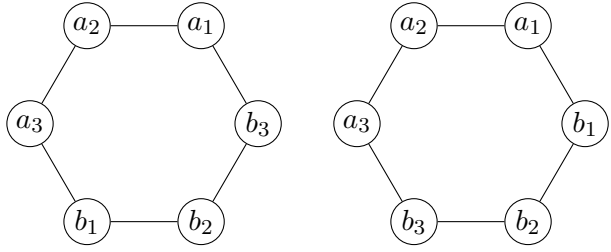
\begin{figure}
    \centering	\begin{tikzpicture}[
	element/.style={shape=circle,draw, fill=white, scale = 0.9}
	]
	\pgfmathsetmacro{\distfromcenter}{1.5}
	\pgfmathsetmacro{\shifttostable}{4.3}
	
	\foreach \i in {1,2,3,4,5,6}{
        \pgfmathsetmacro{\ang}{60*\i}
		\node[element] (v\i) at (\ang:\distfromcenter) {\color{white}$0$};
		}

		\node at (v1) {$a_1$};
		\node at (v2) {$a_2$};
		\node at (v3) {$a_3$};
		\node at (v4) {$b_1$};
		\node at (v5) {$b_2$};
		\node at (v6) {$b_3$};

		\foreach \k in {1,2,3,4,5}
		{
		\pgfmathparse{int(\k+1)}
		\draw (v\k) edge (v\pgfmathresult);
		}
		\draw (v1) edge (v6);
		
	\foreach \i in {1,2,3,4,5,6}{
        \pgfmathsetmacro{\ang}{60*\i}
		\node[element] (w\i) at ($(\shifttostable,0)+(\ang:\distfromcenter)$) {\color{white}$0$};
		}

		\node at (w1) {$a_1$};
		\node at (w2) {$a_2$};
		\node at (w3) {$a_3$};
		\node at (w4) {$b_3$};
		\node at (w5) {$b_2$};
		\node at (w6) {$b_1$};

		\foreach \k in {1,2,3,4,5}
		{
		\pgfmathparse{int(\k+1)}
		\draw (w\k) edge (w\pgfmathresult);
		}
		\draw (w1) edge (w6);

	\end{tikzpicture}
	\caption{Illustration for the proof of \Cref{prop:swapcycle}.  \label{fig:swapcycle}}
\end{figure}

\begin{proposition}
\label{prop:swapcycle}
    For any distance factor function, there exists an instance of TDG and assignments $\asgnm,\asgnm'$ such that every sequence of swaps starting from $\asgnm$ necessarily cycles, but $\asgnm'$ is swap stable.
\end{proposition}

\begin{proof}
For notational convenience, let $N = \{a_i, b_i\colon i\in [3]\}$, and let the topology graph be a cycle of length $6$.
For each $x\in \{a,b\}$ and $i\in [3]$, agent~$x_i$'s utilities are given by $u_{x_i}(x_{i+1}) = 1$ and $u_{x_i}(y) = 0$ for all $y\ne x_{i+1}$, where indices are read modulo~$3$.

Let $\asgnm$ be the assignment on the left of \Cref{fig:swapcycle}. 
The only beneficial swap is between agents $a_3$ and $b_3$. 
Then, in $\asgnm^{a_3 \leftrightarrow b_3}$, the only beneficial swap is between agents $a_2$ and~$b_2$. 
This leads to an assignment where the only beneficial swap is between agents $a_1$ and~$b_1$. 
After these three swaps, we reach the same assignment as $\asgnm$ (up to rotation of the topology graph). 
On the other hand, one can check that the assignment $\asgnm'$ on the right of \Cref{fig:swapcycle} is swap stable.
\end{proof}

\end{document}